\documentclass[aps,pra,twocolumn,10pt,floatfix,aps]{revtex4-1}
\usepackage{amsmath,graphicx, verbatim,enumerate,mathrsfs,mathtools}
\usepackage{bbm}
\usepackage{amssymb}
\usepackage{graphicx}
\usepackage{times}
\usepackage{color}
\usepackage[amsmath,hyperref,standard,thmmarks]{ntheorem}
\usepackage{palatino}
\usepackage[english]{babel}

\renewtheorem{theorem}{Theorem}
\renewtheorem{proposition}{Proposition}

\renewtheorem{lemma}[theorem]{Lemma}
\renewtheorem{corollary}[theorem]{Corollary}
\renewtheorem{definition}[theorem]{Definition}

\renewtheorem{example}[theorem]{Example}
\renewtheorem{remark}[theorem]{Remark}

\usepackage{hyperref}

\newcommand{\lo}[1]{{\color{black} #1}}

\newcommand{\NN}{\mathbb{N}}

\newcommand{\mc}[1]{\mathcal{#1}}

\newcommand{\Tr}{\operatorname{Tr}}

\newcommand{\PPP}{\boldsymbol{P}_0}
\newcommand{\QQQ}{\boldsymbol{P}_\text{f}}
\newcommand{\PP}{P_0}
\newcommand{\QQ}{P_\text{f}}

\begin{document}
\title{Nonlocality free wirings and the distinguishability between Bell boxes}

\author{Rodrigo Gallego}
\affiliation{Dahlem Center for Complex Quantum Systems, Freie Universit\"at Berlin, 14195 Berlin, Germany}
\author{Leandro Aolita}
\affiliation{Instituto de F\'isica, Universidade Federal do Rio de Janeiro, P. O. Box 68528, Rio de Janeiro, RJ 21941-972, Brazil}

\begin{abstract}
Bell nonlocality can be formulated in terms of a resource theory with local-hidden variable models as resourceless objects. Two such theories are known, one built upon local operations assisted by shared randomness (LOSRs) and the other one allowing, in addition, for prior-to-input classical communication.
We show that prior communication, although unable to create nonlocality, leads to wirings not only beyond LOSRs but also not contained in a much broader class of (nonlocality-generating) global wirings. Technically, this is shown by proving that it can improve the statistical distinguishability between Bell correlations optimised over all fixed measurement choices. This has implications in nonlocality quantification, and leads us to a natural universal definition of Bell nonlocality measures. To end up with, we also consider the statistical strength of nonlocality proofs. We point out some issues of its standard definition in the resource-theoretic operational framework, and suggest simple fixes for them. Our findings reveal non-trivial features of the geometry of the set of wirings and may have implications in the operational distinguishability of nonlocal behaviors.
\end{abstract}
\maketitle
\section{Introduction}

Bell nonlocality is an exotic quantum phenomenon by which correlations between the outcomes of space-like separated measurements cannot be explained by local hidden-variable theories, i.e., by any classical model relaying exclusively on past common causes \cite{Bell64,Brunner13}. Apart from their fundamental implications, such nonlocal correlations have been identified as a valuable resource for practical information-theoretic tasks, such as quantum key distribution \cite{Barrett05, Acin06, Acin07}, perfect-randomness expansion \cite{Colbeck07, Pironio10} and \cite{Colbeck12,Gallego13} amplification, or distributed computations \cite{ Buhrman10}, e.g. 
These protocols exploit the nonlocal correlations shared between distant users as physical resources, and process their generated classical information locally, with the aid of a-priori communicated bits or shared randomness. In these scenarios, Bell nonlocality plays the role of an operational resource  for classical-information processing, i.e., one that can be composed, acted on, and transformed between its different forms, depending on one's needs \cite{Barrett05b, Allcock09, Lang14, Beigi14}. This is formalised by so-called \emph{resource theories of Bell nonlocality} \cite{Gallego12,Joshi13, deVicente14,Horodecki15}.

Resource theories constitute powerful formalisms in quantum information for the abstract treatment of a physical property as an operational resource. They have been built also for other types of quantum nonlocality, such as entanglement \cite{HorodeckiHorodeckiHorodecki} and Einstein-Podolski-Rosen steering \cite{Gallego15}, as well as athermal states \cite{Resourcetheorythermodynamics}, quantum coherence \cite{Streltsov16}, and several other notorious properties of quantum systems (see, e.g., Ref. \cite{Brandao15} and references therein). 
A resource theory identifies a set of mathematical objects -- describing states of the physical system under scrutiny --, a subset thereof of uninteresting objects -- composed of the states without the property considered the resource in question --, and a class of free operations -- consisting of physical transformations under which the resourceless subset is closed. For Bell nonlocality, two classes of free operations are well studied: one restricted to \emph{local operations assisted by shared randomness} (LOSR) \cite{deVicente14,Horodecki15} and the other allowing also for communication that contains no information about the measurement settings of the transformed object, called \emph{wirings and prior-to-input classical communication} (WPICCs) \cite{Gallego12}. Until this work, it was not clear whether these two classes displayed any difference. In fact, they have sometimes been referred to as the same single class \cite{deVicente14,Horodecki15}. Here, we clarify this question showing that the classes are actually inequivalent.

On the other hand, central to any resource theory is the distinguishability between the objects considered, and, in particular, the distinguishability between resourceful and resourceless ones. The \emph{distinguishability between Bell correlations} was studied in Ref. \cite{vanDam05}. There, a measure of statistical distinguishability between the measurement outcomes of a given nonlocal device and any local one, called the statistical strength of nonlocality proofs, was put forward. However, such measure was derived from a game-theoretic perspective and consistency with resource-theoretic operational approaches to nonlocality was not checked for.

Here, we study the interconnection between the two main resource theories of Bell non-locality, based on the paradigms of LOSRs and WPICCs, in view of the operational task of distinguishing Bell boxes. To begin with, we derive an explicit parametrisation for a generic WPICC map. Then, we show that WPICCs can be used to increase the statistical distinguishability between two Bell boxes. This is, in contrast, impossible not only with LOSRs, but also with the more general \emph{global wirings} (GWs), defined analogously to LOSRs but with arbitrary nonlocal correlations playing the role of shared randomness. Technically, this is proven by showing that the relative entropy between Bell-box behaviors increases under WPICCs, while it can only decrease under all GWs. 

We then move on to study the quantification of Bell nonlocality with respect to the two different resource theories. In particular, we show that every nonlocality measure with respect to LOSRs is also a valid measure with respect to WPICCs, and that the converse holds if the quantifier satisfies the natural requirement of convexity. This leads us to a universal definition of Bell non-locality monotones, consistent with both resource theories. As an example of such monotone, we provide a definition of the relative entropy of nonlocality in terms of the relative entropy between behaviors. Finally, we probe the three variants, introduced in Ref. \cite{vanDam05}, of the statistical strength of nonlocality proofs as Bell nonlocality monotones. One of them coincides with the relative entropy of nonlocality and is, therefore, automatically a nonlocality monotone. We prove that, from the other two, one is a Bell nonlocality monotone whereas the other one is not even monotonous under LOSRs. We end up by providing physical arguments by which, from the two variants defining satisfactory non-locality measures, we find one better than the other as quantifier of the statistical strength of nonlocality proofs, even from a game-theoretic perspective.

The paper is structured as follows. In Sec. \ref{sec:Basics} we introduce the basic notions and notation. In Sec. \ref{sec:Op_Framework} we discuss the different classes of transformations between Bell boxes. An explicit analytic expression for generic WPICCs is provided there too. In Sec. \ref{sec:rel_ent} we prove that the relative entropy between behaviors can increase under WPICCs, while it can only decrease under GWs. In Sec. \ref{sec:nl_monotone} we study generic Bell nonlocality monotones and, in particular, the relative entropy of nonlocality.
In Sec. \ref{sec:st_strength} we revisit statistical strength of nonlocality proofs from a resource-theoretic perspective. We finish the paper in Sec. \ref{sec:conclu} with a few relevant final remarks. 

\section{Preliminaries}
\label{sec:Basics}
In this section we introduce the basics of Bell nonlocality \cite{Brunner13}. We consider two space-like separated experimenters, Alice and Bob, who make local measurements on a bipartite system composed of two black-box devices. Alice's box has $x\in[s_A]$ measurement settings (inputs) and $a\in[r_A]$ measurement results (outputs) and, similarly, Bob's box admits $y\in[s_{\rm{b}}]$ inputs and returns $b\in[r_B]$ outputs, where $s_A,\ r_A,\ s_{\rm{b}},\ r_B \in\NN$, and the notation $[n]\coloneqq\{0,\ldots,n-1\}$, for any $n \in\mathbb{N}$, has been introduced. For notational simplicity, but without loss of generality, we take $s\coloneqq s_A=\ S_{\rm{b}}$ and $r\coloneqq r_A=r_B$. 
The experiment is described by a normalised bipartite conditional probability distribution \begin{equation}
\label{def:Prob_dist}
\boldsymbol{P}\coloneqq\{P(a,b|x,y)\}_{a,b\in[r],\, x,y\in[s]},
\end{equation}
where $P(a,b|x,y)$ is the conditional probability of obtaining the output values $a$ and $b$ given that the input values $x$ and $y$. We refer to any normalised bipartite conditional probability distribution as a \emph{black-box behavior} or, simply, \emph{behavior}, for short.

Since the measurements constitute space-like separated events, the statistics must fulfil the no-signaling principle, given by the well-known conditions
\begin{subequations}
\label{eq:nosignaling}
\begin{align}
P(b|y)=&\sum_{a'}P(a',b|x',y) \:\:\: \forall\ x'\in[s],\\
P(a|x)=&\sum_{b'}P(a,b'|x,y') \:\:\: \forall\ y'\in[s],
\end{align}
\end{subequations}
for all $a,b\in[r]$ and $x,y\in[s]$. That is, the marginal conditional distribution  $\{P(a|x,y)\}_{a\in[r],\, x,y\in[s]}$ ($\{P(a|x,y)\}_{b\in[r],\, y\in[s]}$) for Alice (Bob) should not depend on Bob's (Alice's) measurement choice.  We refer to the set of all behaviors that fulfil the linear constraints \eqref{eq:nosignaling} as $\mathsf{NS}$. In addition, we denote by $\mathsf{Q}$ the set of all quantum behaviors, i.e., all those that can be expressed as 
\begin{align}
\label{eq:Q}
P(a,b|x,y)=\Tr\left[\varrho_{AB}\,M_x^a\otimes M_y^b\right],
\end{align}
where $\varrho_{AB}$ is a bipartite quantum state and $M_x^a$ and $M_y^b$ are local measurement operators corresponding to Alice and Bob, respectively. In turn, we call $\mathsf{L}$ the set of all local behaviors, i.e., all those for which there exists a normalised probability distribution $\boldsymbol{P}_{\Lambda}$ and two normalised conditional probability distributions $\boldsymbol{P}_{A|X,\Lambda}$ and $\boldsymbol{P}_{B|Y,\Lambda}$ such that
\begin{align}
\label{eq:LHV}
P(a,b|x,y)=\sum_{\lambda}P_{\Lambda}(\lambda)\, P_{A|X,\Lambda}(a|x,\lambda)\, P_{B|Y,\Lambda}(b|y,\lambda),
\end{align}
for all $a,b\in[r]$ and $x,y\in[s]$. The variable $\lambda$ is called a local-hidden variable and the decomposition \eqref{eq:LHV} is accordingly referred to as a local-hidden variable (LHS) model for the behavior. It is a well-known fact that 
\begin{equation}
\label{eq:L_subset_Q_subset_NS}
\mathsf{L} \subset \mathsf{Q}\subset\mathsf{NS}.
\end{equation}
Finally, we say that any $\boldsymbol{P}\notin\mathsf{L}$ is a \emph{nonlocal} behavior, and refer to this fact as \emph{nonlocality}.

\begin{figure*}[t!]
\centering
\includegraphics[width=1\linewidth]{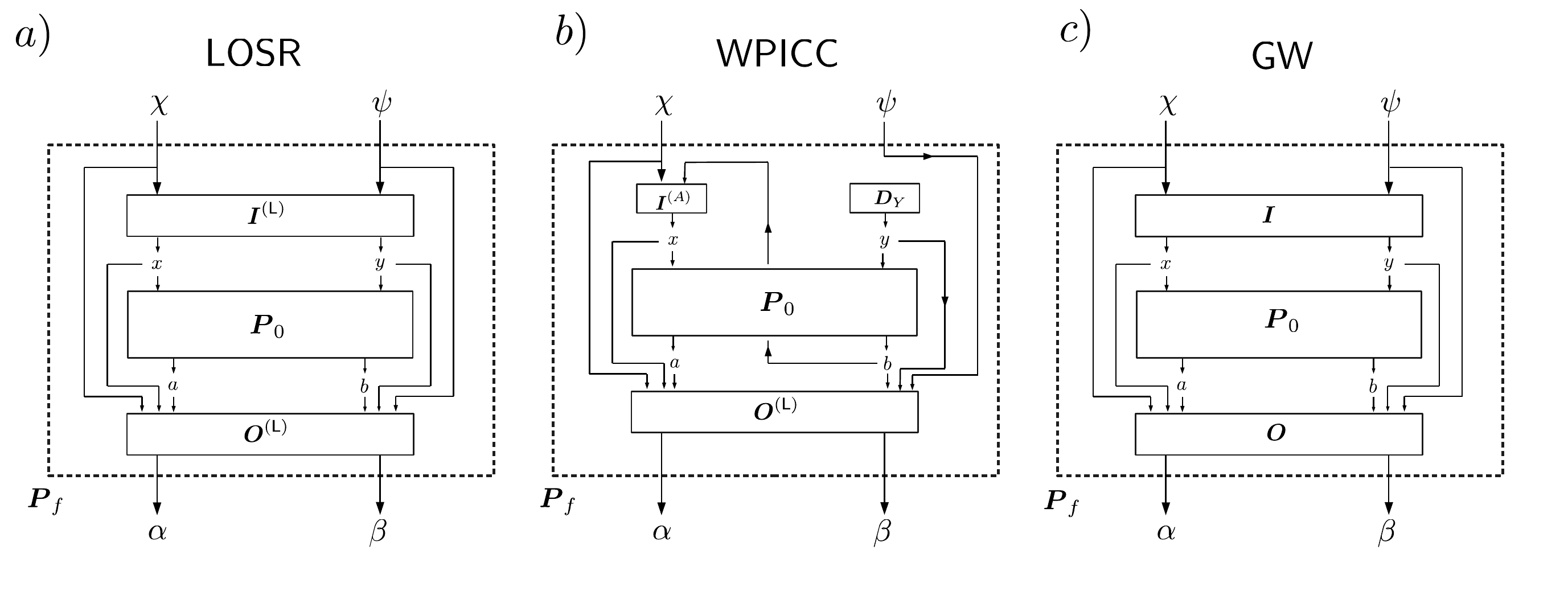}
\caption{\textbf{Circuit representation of the three main types of Bell-box wirings}. Panel $a$) shows a generic $\mathsf{LOSR}$ wiring. There, the inputs of the final box $\QQQ$ are processed by an input box with local behavior $\boldsymbol{I}^{(L)}$, whose outputs are input to the initial box $\PPP$. The outputs produced by $\PPP$ are, in turn, processed by an output box with local behavior $\boldsymbol{O}^{(L)}$, which has as inputs all previously generated dits, without any exchange of dits between Alice and Bob's sides. The outputs of $\boldsymbol{O}^{(L)}$ are the final outputs. Panel $b$) shows an example of the class $\mathsf{WPICC}$, which allows for communication between the users  provided it does not carry any information about the final inputs. In the example, Bob measures his initial local box before his final input is decided. He chooses his initial input with a single-partite box $\boldsymbol{D}_Y$ without inputs. The output of his initial box is sent to Alice, who uses it, together with her final input, to choose her initial input according to a single-partite behavior $\boldsymbol{I}^{(A)}$. The remaining dit processing is the same as in a $\mathsf{LOSR}$ wiring. Finally, in panel $c$), a generic $\mathsf{GW}$ wiring is displayed. The circuitry there is analogous to that of $\mathsf{LOSR}$ wirings, except that the input and output boxes are governed by generic nonlocal behaviors $\boldsymbol{I}$ and $\boldsymbol{O}$, respectively, not restricted to $\mathsf{L}$. In fact, $\boldsymbol{I}$ and $\boldsymbol{O}$ need in general not even be in $\mathsf{NS}$. 
Important sub-classes of $\mathsf{GW}$ are the no-signaling and quantum wirings, obtained when $\boldsymbol{I}$ and $\boldsymbol{O}$ are restricted to $\mathsf{NS}$ and $\mathsf{Q}$, respectively.
While $\mathsf{LOSR}$ and $\mathsf{WPICC}$ are classes of nonlocality-free wirings, $\mathsf{GW}$ can create Bell nonlocality, i.e., it can map local behaviors into nonlocal ones. }
\label{fig:1}
\end{figure*}

\section{Nonlocality as an operational resource}
\label{sec:Op_Framework}
Here, we focus on nonlocality as an operational resource for information-theoretic tasks. This is formalised by so-called resource theories \cite{Gallego12,deVicente14,Horodecki15,Gallego15}. Resource theories are composed of three main elements: {\it i)} \emph{mathematical objects} describing the system under scrutiny; \emph{ii)} a particular property of such objects considered the valuable \emph{resource}; and \emph{iii)} a class of \emph{free operations} for the resource, i.e. physical transformations fulfilling the essential requirement of mapping all free objects (i.e., those without the resource) into free objects. For Bell nonlocality, these three components are: 
\begin{enumerate}[{\it i)}]
\item behaviors $\boldsymbol{P}\in\mathsf{NS}$;
\item nonlocality, i.e., that $\boldsymbol{P}\notin \mathsf{L}$;
\item physical transformations under which $\mathsf{L}$ (the set of free objects) is closed, i.e., that do not create nonlocality.
\end{enumerate}
Importantly, the requirement that the free operations leave $\mathsf{L}$ invariant is necessary but not sufficient to specify the concrete class of free operations. Typically, this specification is made on the basis of the physical restrictions native of the scenario where nonlocality serves as a resource. Consequently, in general, there can be multiple classes of free operations, and, therefore, of resource theories, for the same resource, as we discuss next.

\subsection{Two resource theories of Bell nonlocality}
\label{subsec:wirings}
We restrict throughout to the paradigm of linear maps from behaviors into behaviors. Each such transformation can be physically  realised by wiring inputs and outputs of the initial black-box measurement devices with the inputs and outputs of other black boxes \cite{Short, Barrett}. Hence, from now on, we refer to any linear black-box transformation as a \emph{wiring}. 

Within a resource theory of nonlocality \cite{deVicente14,Horodecki15},  a wiring $\mathcal{W}$ is a free operation for  nonlocality if 
\begin{align}
\label{eq:free_ops}
\QQQ=\mathcal{W}\left(\PPP\right)\in\mathsf{L},\ \forall\ \boldsymbol{P}\in\mathsf{L}.
\end{align}
$\QQQ\coloneqq\{\QQ(\alpha,\beta|\chi,\psi)\}_{\alpha,\beta\in[r_f],\, \chi,\psi\in[s_f]}$ represents the final behavior after the transformation $\mathcal{W}$ on an initial behavior $\PPP$ of the form \eqref{eq:LHV}, where $\QQ(\alpha,\beta|\chi,\psi)$ is the conditional probability of obtaining the output values $\alpha$ and $\beta$ given that the input values $\chi$ and $\psi$ for the final box. Note that, in full generality, we allow the cardinality of the alphabets of inputs and outputs to change (from $s$ and $r$ to $s_f$ and $r_f$, respectively). We refer to any wiring that is a free operation for nonlocality as a \emph{nonlocality-free wiring}, or, for short, simply a \emph{free wiring}. There are two non-trivial classes of free wirings known.

\subsubsection{Local operations and shared-randomness}
The first class is called \emph{local operations assisted by shared randomness}  (see, e.g., Ref. \cite{deVicente14} for a review on the topic). This class, which we denote as $\mathsf{LOSR}$, encodes the physical restriction that Alice and Bob can only process  the classical information available \emph{locally}, without any communication between them. It is composed of all wirings $\mathcal{W}_\mathsf{LOSR}$ explicitly parametrised by  $\QQQ=\mathcal{W}_\mathsf{LOSR}\left(\PPP\right)$, with 
\begin{align}
\label{eq:LOSR}
\nonumber
P_\text{f}(\alpha,\beta\vert  \chi,\psi):=
\sum_{a,b,x,y}&O^{(\mathsf{L})}(\alpha,\beta\vert a,b,x,y,\chi,\psi)\times\\
&\PP(a,b\vert x,y)\times I^{(\mathsf{L})}(x,y\vert \chi,\psi),
\end{align}
where $\boldsymbol{I}^{(\mathsf{L})}$ and $\boldsymbol{O}^{(\mathsf{L})}$ are arbitrary boxes in $\mathsf{L}$. That is, the input and output dits of the initial  box $\PPP$ are processed (wired) locally, as sketched in Fig. \ref{fig:1} a), with (well-normalised) input and output behaviors $\boldsymbol{I}^{(\mathsf{L})}$ and $\boldsymbol{O}^{(\mathsf{L})}$ that admit both a LHV model, to produce the input and output dits of the final box $\QQQ$. It is known facts that if $\PPP\in\mathsf{L}$ then $\mathcal{W}_\mathsf{LOSR}(\PPP)\in\mathsf{L}$ and if $\PPP\in\mathsf{NS}$ then $\mathcal{W}_\mathsf{LOSR}(\PPP)\in\mathsf{NS}$. Examples of wirings in $\mathsf{LOSR}$ are local relabelings of inputs or outputs and mixing with a local behavior \cite{deVicente14}.
%

\subsubsection{Wirings and prior to inputs classical communication}

The second class is commonly known as \emph{wirings and prior-to-input classical communication} \cite{comment1}. This is the class $\mathsf{WPICC}$ of all wirings $\mathcal{W}_\mathsf{WPICC}$ operationally defined by the following two-stage sequence (see Fig. \ref{fig:1} b) \cite{Gallego12}.
\begin{enumerate}
\item \emph{Preparation phase:} Alice and Bob are allowed to use the initial box, i.e., to choose $x$ and/or $y$, generating $a$ and/or $b$, respectively, and to communicate $x$, $y$, $a$, $b$, or any other random bit of their choice, before the final inputs $\chi$ and $\psi$ are chosen.
\item \emph{Measurement phase:} Once $\chi$ and $\psi$ are chosen, Alice and Bob apply a generic  $\mathsf{LOSR}$ wiring. 
\end{enumerate}
This sequence unambiguously defines the class $\mathsf{WPICC}$. In App. \ref{sec:Par_WPICC}, we provide an explicit parametrisation of generic wirings in the class.
The analytic expression obtained is somewhat cumbersome due to the many different options that branch off during the preparation phase, but it can be written in a simplified form as
\begin{equation}
\label{eq:finalWPICC}
\mathcal{W}_\mathsf{WPICC}(\PPP)\coloneqq p\,L(\PPP)+(1-p)\,\mathcal{W}_\mathsf{LOSR}(\PPP),
\end{equation}
where $0\leq p\leq 1$, $\mathcal{W}_\mathsf{LOSR}\in\mathsf{LOSR}$, and $L$ is a linear map from $\mathsf{NS}$ to $\mathsf{L}$, i.e., $L(\PPP)\in\mathsf{L}$ for all $\PPP\in\mathsf{NS}$. From Eq. \eqref{eq:finalWPICC}, one immediately sees that if $\PPP\in\mathsf{L}$ then $\mathcal{W}_\mathsf{WPICC}(\PPP)\in\mathsf{L}$ and if $\PPP\in\mathsf{NS}$ then $\mathcal{W}_\mathsf{WPICC}(\PPP)\in\mathsf{NS}$.  Besides, it clearly holds that 
\begin{equation}
\label{eq:LOSR_subset_WPICC}
\mathsf{LOSR}\subseteq\mathsf{WPICC}.
\end{equation}
However, as we show in Corolary \ref{Col:non_inclu} in Sec. \ref{sec:WPICC_out_of_GW}, the reciprocal turns out not to be true.

As a final but important remark, we note that, if $\PPP\notin \mathsf{NS}$, causal loops can arise due to the preparation phase. For instance, in the second panel of Fig. \ref{fig:1}, the initial output $b$ is used to choose the initial input $x$. This can clearly introduce a causal loop if $\PPP$ is signaling from Alice to Bob. If, on the contrary, $\PPP\in \mathsf{NS}$, such problems are avoided and the preparation phase is consistent. Hence, throughout, we restrict the domain of $\mathsf{WPICC}$ wirings to the set $ \mathsf{NS}$ of no-signaling  behaviors. See  App. \ref{sec:Par_WPICC} for more details about  $\mathsf{WPICC}$.

\subsection{Nonlocal wirings}
Next, we consider a third class of wirings, which we call \emph{global wirings} and denote by $\mathsf{GW}$. $\mathsf{GW}$ is not a class of free operations for Bell nonlocality, but it is relevant to the results we discuss below. 
The class is composed of all the wirings $\mathcal{W}_\mathsf{GW}$ that act globally on the input and output dits, without any restriction of locality or even no-signaling. They process $\PPP$ as an effective single-partite distribution with the inputs $(x,y)$ and outputs $(a,b)$ treated as higher-dimensional single-partite inputs and outputs, respectively. 
They are defined as the $\mathsf{LOSR}$ wirings but with generic (instead of local) boxes wired to the input and output dits of the initial box, explicitly parametrised by  $\QQQ=\mathcal{W}_\mathsf{GW}\left(\PPP\right)$, with 
\begin{align}
\label{eq:GO}
\nonumber
P_\text{f}(\alpha,\beta\vert  \chi,\psi):=
\sum_{a,b,x,y}&O(\alpha,\beta\vert a,b,x,y,\chi,\psi)\times\\
&\PP(a,b\vert x,y)\times I(x,y\vert \chi,\psi),
\end{align}
with $\boldsymbol{I}$ and $\boldsymbol{O}$ arbitrary (possibly even signaling) boxes. By construction, it clearly holds that 
\begin{equation}
\label{eq:LOSR_subset_GW}
\mathsf{LOSR} \subset \mathsf{GW}.
\end{equation}
Furthermore, for any two arbitrary behaviors $\PPP$ and $\QQQ$ there exists $\mathcal{W}_\mathsf{GW} \in \mathsf{GW}$ such that $\QQQ=\mathcal{W}_\mathsf{GW}(\PPP)$. A simple way to see this is by constructing a global wiring $\mathcal{W}_\mathsf{GW}$ that bypasses $\PPP$ and directly generates $\QQQ$, i.e., by taking $\boldsymbol{O}(\cdot,\cdot|\alpha,\beta,\chi,\psi,\cdot,\cdot)=\QQQ$ for all $\alpha$, $\beta$, $\chi$, and $\psi$. Clearly, such $\mathcal{W}_\mathsf{GW}$ can map initial local boxes to arbitrary (no-signaling as well as signaling) final boxes.
However, as we see in Sec. \ref{sec:WPICC_out_of_GW}, surprisingly, this does not imply that $\mathsf{GW}$ contains all physical wirings. 

Finally, since $\mathsf{GW}$ wirings can map no-signaling behaviors out of $\mathsf{NS}$, and since $\mathsf{WPICC}$ wirings are well-defined only on behaviors in $\mathsf{NS}$, $\mathsf{GW}$ wirings cannot in general be composed with $\mathsf{WPICC}$ wirings. This suggests to consider a fourth class: the \emph{no-signaling wirings}, which we denote by $\mathsf{NSW}$. This class comprises all global wirings for which $\boldsymbol{I},\boldsymbol{O}\in\mathsf{NS}$, i.e., it is defined also by Eq. \eqref{eq:GO} but with the restriction that the boxes with which the inputs and outputs of the initial behavior are wired are described by no-signaling  distributions. The class $\mathsf{NSW}$ is a highly relevant in a variety of physical scenarios. Nevertheless, below, we prove our results directly for the superset $\mathsf{GW}$. The validity of our results for the subset $\mathsf{NSW}$ is  automatic by inclusion. 
\section{Prior classical communication improves box distinguishability}
\label{sec:rel_ent}
In this section, we study the inclusion relationships between the different classes of wirings. We show that the set $\mathsf{WPICC}$ is strictly larger than $\mathsf{LOSR}$. Even more surprisingly, we also show that $\mathsf{WPICC}$ is not contained in $\mathsf{GW}$. Far from being a mere mathematical curiosity, we show the implications of these inclusions in the operational task of distinguishing black boxes.

\subsection{The relative entropy between behaviors as a measure of their distinguishability}
\label{sec:cond_rel_ent}
Consider two arbitrary behaviors with equal alphabets of inputs and outputs: $\boldsymbol{P}$, given by Eq. \eqref{def:Prob_dist}, and $\boldsymbol{P}'$, given by 
\begin{eqnarray}
\label{def:Prob_dist_prime}
\boldsymbol{P}'\coloneqq\{P'(a,b|x,y)\}_{a,b\in[r],\, x,y\in[s]}.
\end{eqnarray}
Imagine next that we wish the distinguish them by choosing their inputs according to a generic joint probability distribution $\boldsymbol{D}\coloneqq\{D(x,y)\}_{x,y\in[s]}$ and then comparing the resulting overall  input-output statistics, i.e.: 
\begin{subequations}
\label{eq:input_output_dists}
\begin{align}
\boldsymbol{P}\cdot \boldsymbol{D}\coloneqq&\{P(a,b|x,y)\, D(x,y)\}_{a,b\in[r],\, x,y\in[s]}\\
\intertext{and}
\boldsymbol{P}'\cdot \boldsymbol{D}\coloneqq&\{P'(a,b|x,y)\, D(x,y)\}_{a,b\in[r],\, x,y\in[s]}.
\end{align}
\end{subequations}
The distinguishability between the two probability distributions can be quantified by the \emph{relative entropy} (RE) $S$, also known as the \emph{Kullback-Leibler divergence} \cite{KullbackLeibler}. More precisely, for any two distributions $\boldsymbol{Q}\coloneqq\{Q(z)\}_{z}$ and $\boldsymbol{Q}'\coloneqq\{Q'(z)\}_{z}$, the RE
of $\boldsymbol{Q}$ with respect to $\boldsymbol{Q}'$ is defined as 
\begin{align}
\label{eq:def_KLD}
S\left(\boldsymbol{Q}\|\boldsymbol{Q}'\right)\coloneqq&\sum_{z}Q(z)\log\left(\frac{Q(z)}{Q'(z)}\right).
\end{align}
$S$ is the most-widely accepted measure of distinguishability between two probability distributions \cite{KullbackLeibler,Chernoff52,Blahut1974,Audenaert2012,Cover06}. In the asymptotic infinite-sample scenario \cite{KullbackLeibler,Chernoff52,Blahut1974}, $S\left(\boldsymbol{Q}\|\boldsymbol{Q}'\right)$ quantifies the statistical confidence that a sample $z$ generated by $\boldsymbol{Q}$ gives, on average, in favour of the hypothesis that the data have indeed been sampled from $\boldsymbol{Q}$ and against the hypothesis that the data have been produced by $\boldsymbol{Q}'$.

Thus, the RE 
\begin{align}
\label{eq:cond_RE}
\nonumber
S\left(\boldsymbol{P}\cdot \boldsymbol{D}\|\boldsymbol{P}'\cdot\boldsymbol{D}\right)\coloneqq&\\
\nonumber
\sum_{a,b\in[r],\, x,y\in[s]}P(a,b|x,y)\, D(x,y)\times \log\left(\frac{P(a,b|x,y)}{P'(a,b|x,y)}\right)=&\\
\sum_{x,y\in[s]}D(x,y) S\left(\boldsymbol{P}(\cdot,\cdot|x,y)\|\boldsymbol{P}'(\cdot,\cdot|x,y)\right)&
\end{align}
of $\boldsymbol{P}\cdot \boldsymbol{D}$ with respect to $\boldsymbol{P}'\cdot \boldsymbol{D}$ measures the average distinguishability between the outputs of $\boldsymbol{P}$ and  $\boldsymbol{P}'$ when the inputs are chosen according to the common distribution $\boldsymbol{D}$. In Eq. \eqref{eq:cond_RE}, $\boldsymbol{P}(\cdot,\cdot|x,y)$ and $\boldsymbol{P}'(\cdot,\cdot|x,y)$ stand for the probability distributions over the outputs obtained from $\boldsymbol{P}$ and $\boldsymbol{P}'$, respectively, for a fixed choice of inputs $(x,y)$.
%

This motivates a very reasonable definition for the relative entropy between behaviors, namely, by optimising the overall distinguishability $S\left(\boldsymbol{P}\cdot \boldsymbol{D}\|\boldsymbol{P}'\cdot\boldsymbol{D}\right)$ over all possible input distributions $\boldsymbol{D}$:
\begin{definition}[behavior RE] 
\label{dfn:rel_ent_b}
The relative entropy $S_{\rm{b}}\left(\boldsymbol{P}\|\boldsymbol{P}'\right)$ of $\boldsymbol{P}$ with respect to $\boldsymbol{P}'$ is defined by 
\begin{eqnarray}
\nonumber 
S_{\rm{b}}\left(\boldsymbol{P}\|\boldsymbol{P}'\right)&\coloneqq&\max_{\boldsymbol{D}}S\left(\boldsymbol{P}\cdot \boldsymbol{D}\|\boldsymbol{P}'\cdot\boldsymbol{D}\right)\\
\label{eq:relent_behavior}
&=&\max_{x,y\in[s]} S\left(\boldsymbol{P}(\cdot,\cdot|x,y)\|\boldsymbol{P}'(\cdot,\cdot|x,y)\right).
\end{eqnarray}
\end{definition}
Equality \eqref{eq:relent_behavior} follows immediately from Eq. \eqref{eq:cond_RE} and the positive semi-definiteness of $S\left(\boldsymbol{P}(\cdot,\cdot|x,y)\|\boldsymbol{P}'(\cdot,\cdot|x,y)\right)$. $S_{\rm{b}}\left(\boldsymbol{P}\|\boldsymbol{P}'\right)$ thus measures the statistical distinguishability between the output probability distributions given by the behaviors $\boldsymbol{P}$ and $\boldsymbol{P}'$ when their inputs are fixed at the optimal values that maximise the output distinguishability in question. 

The RE between probability distributions is known to be contractive -- i.e., non-increasing -- under all linear maps between probability distributions. Contractivity under physical transformations is a property of uttermost importance for any reasonable measure of distinguishability. For instance, the RE between quantum states is known to be contractive under generic completely positive maps  \cite{Hellstroem76, Holevo78}, whereas the RE between steering assemblages is known to be contractive under the free operations of steering \cite{Gallego15}. In App. \ref{sec:Proof_theor:relentmonotone}, we prove the analogous for the behavior RE under  $\mathsf{GW}$ wirings:
\begin{theorem}(Contractivity of $S_{\rm{b}}$ under $\mathsf{GW}$)
\label{theor:relentmonotone} 
Let $\boldsymbol{P},\, \boldsymbol{P}'\in\mathsf{NS}$ be any two no-signaling behaviors. Then
\begin{equation}\label{eq:relentmonotone}
S_{\rm{b}}\left(\mathcal{W}_\mathsf{GW}(\boldsymbol{P})\|\mathcal{W}_\mathsf{GW}(\boldsymbol{P}')\right) \leq S_{\rm{b}}\left(\boldsymbol{P}\|\boldsymbol{P}'\right)
\end{equation}
for all $\mathcal{W}_\mathsf{GW}\in \mathsf{GW}$.
\end{theorem}
\noindent Note that, by the inclusion relationship \eqref{eq:LOSR_subset_GW}, Thm. \ref{theor:relentmonotone} automatically implies that $S_{\rm{b}}$ is contractive under $\mathsf{LOSR}$ wirings. Consequently, if only $\mathsf{LOSR}$ or  $\mathsf{GW}$ wirings are considered, $S_{\rm{b}}$ can be taken as a physically reasonable measure of distinguishability between behaviors. In the next section, we show that, surprisingly, this does not hold for $\mathsf{WPICC}$ wirings.

\subsection{$\mathsf{WPICC}$ outperforms $\mathsf{GW}$ at distinguishing behaviors}
\label{sec:WPICC_out_of_GW}

The RE between behaviors proves additionally a useful tool to assess the relationship between the different classes of wirings. From the study of $S_{\rm{b}}$, we discover an unexpected inequivalence between $\mathsf{WPICC}$ and $\mathsf{GW}$. For that, we first realise the following surprising fact.
\begin{theorem}(Non-contractivity of $S_{\rm{b}}$ under $\mathsf{WPICC}$)
\label{theor:notmonotonewpicc}
There exist wirings $\mathcal{W}_\mathsf{WPICC}\in\mathsf{WPICC}$ and behaviors $\boldsymbol{P},\,\boldsymbol{P}'\in\mathsf{NS}$ such that
\begin{equation}
\label{eq:violatemonotonicity}
S_{\rm{b}} \left(\mathcal{W}_\mathsf{WPICC}(\boldsymbol{P})\|\mathcal{W}_\mathsf{WPICC}(\boldsymbol{P}')\right) >S_{\rm{b}}\left(\boldsymbol{P}\|\boldsymbol{P}'\right)
\end{equation}
\end{theorem}

\noindent The theorem is proven in App. \ref{sec:Proof_theor:notmonotonewpicc} by explicit example construction. An immediate implication is that the class $\mathsf{WPICC}$ is not only not equivalent to $\mathsf{LOSR}$ but also it is not even contained in $\mathsf{GW}$. This follows as a corollary of Thm. \ref{theor:notmonotonewpicc} together with Thm. \ref{theor:relentmonotone}.
\begin{corollary}(Non-inclusion of $\mathsf{WPICC}$ in $\mathsf{GW}$)
\label{Col:non_inclu}
There exist wirings $\mathcal{W}_\mathsf{WPICC} \in \mathsf{WPICC}$ such that $\mathcal{W}_\mathsf{WPICC}\notin \mathsf{GW}$. That is,
\begin{equation}
\label{eq:WPICC_notsubset_GW}
\mathsf{WPICC} \not\subset\mathsf{GW}.
\end{equation}
\end{corollary}
The corollary reveals a very unexpected feature of the internal geometry of the set of wirings, schematically depicted in Fig. \ref{fig:2}.
\begin{figure}[t!]
\centering
\includegraphics[width=.7\linewidth]{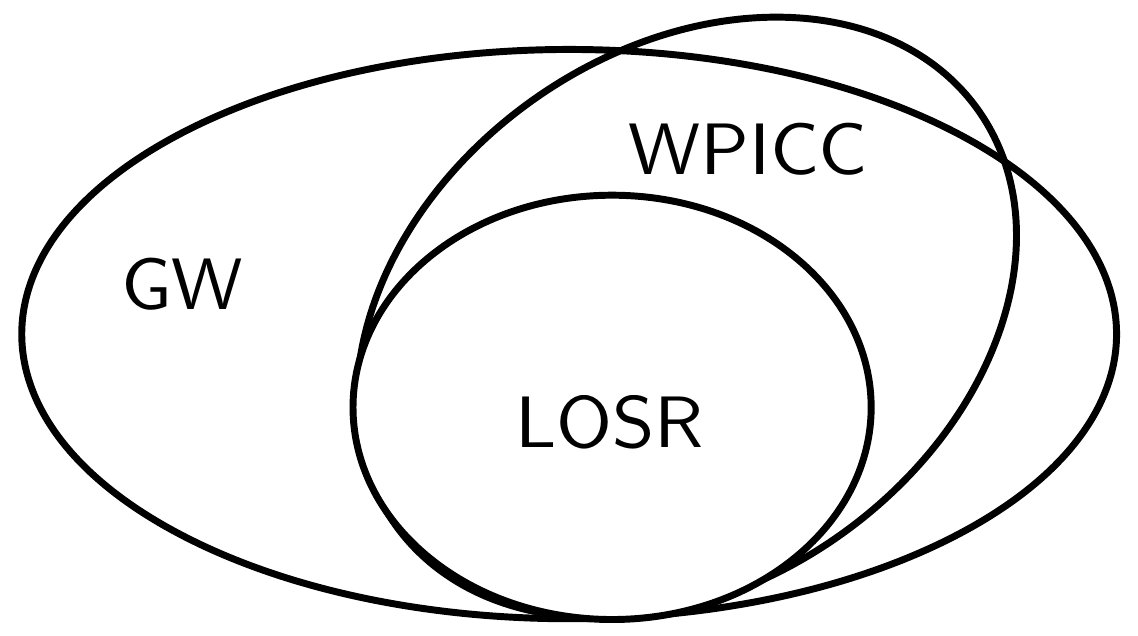}
\caption{\textbf{Inclusion relationships among the three classes of wirings}. In spite of being nonlocality free, some wirings with prior-to-input classical communication are out of the set of global wirings, which are, in general, not nonlocality free.}
\label{fig:2}
\end{figure}

Before we finish this section, let us shortly elaborate on the physical implications of Theorem \ref{theor:notmonotonewpicc} in the information-theoretic task of distinguishing two behaviors $\boldsymbol{P}$ and $\boldsymbol{P}'$.
It is clear that $S_{\rm{b}}$ satisfactorily quantifies the average distinguishability between output samples of $\boldsymbol{P}$ and $\boldsymbol{P}'$ for the optimal choice of fixed inputs. For this task, Theorem \ref{theor:notmonotonewpicc} implies that it is sometimes better to first apply a $\mathsf{WPICC}$ wiring $\mc{W}_{\mathsf{WPICC}}$ to the behaviors and only then choose the optimal inputs (the ones maximising the distinguishability between output samples of $\mc{W}_{\mathsf{WPICC}}(\boldsymbol{P})$ and $\mc{W}_{\mathsf{WPICC}}(\boldsymbol{P}')$, instead of $\boldsymbol{P}$ and $\boldsymbol{P}'$). So, one clearly concludes that $\mathsf{WPICC}$ outperforms $\mathsf{LOSR}$, and even $\mathsf{GW}$, at distinguishing Bell boxes for the (restricted) input-choice strategy in question. 
However, Theorem \ref{theor:notmonotonewpicc} also tells us,  on the other hand, that $S_{\rm{b}}$, as defined in Def. \ref{dfn:rel_ent_b}, cannot be considered a satisfactory measure of the overall distinguishability between $\boldsymbol{P}$ and $\boldsymbol{P}'$ (under generic input-choice strategies). As mentioned above, an essential requirement for a \emph{bona fide} measure of distinguishability is that it does not increase under physical transformations. As implied by Thm. \ref{theor:relentmonotone}, $S_{\rm{b}}$ fulfils this requirement when $\boldsymbol{P}$ and $\boldsymbol{P}'$ are treated as monopartite objects -- i.e., with the pair $(x,y)$ seen as a single-partite input chosen before any output is generated. However, it does not when the fact that $\boldsymbol{P}$ and $\boldsymbol{P}'$ are bipartite objects is explicitly exploited, namely, e.g., when the input to one of the boxes is chosen depending on the output of the other box.

The previous discussion is particularly relevant in scenarios where one wishes to estimate the statistical confidence that nonlocal behaviors give as nonlocality proofs \cite{vanDam05}, i.e., the average distinguishability between a given nonlocal behavior and any behavior in $\mathsf{L}$. In the literature, $S_{\rm{b}}$ has been employed as the canonical measure of that statistical strength. Nonetheless, the analysis above puts the canonical approach into question and suggests to study distinguishability measures from a resource-theoretic perspective. This is what we do in the next two sections.

\section{Bell nonlocality monotones}
\label{sec:nl_monotone}
The basic necessary condition for a function to be a satisfactory measure of Bell nonlocality from a resource-theoretic viewpoint is that it is non-increasing under the free operations for Bell nonlocality  \cite{Gallego12,deVicente14,Horodecki15}. For the free wirings, this is formalised by the following definition.
\begin{definition}($\mathsf{LOSR}$ and $\mathsf{WPICC}$ monotones)
\label{def:monot} 
A function $f:\mathsf{NS}\to\mathbb{R}_{\geq0}$ is an \emph{$\mathsf{LOSR}$ ($\mathsf{WPICC}$) monotone} if:
\begin{enumerate}[i)]
\item $f(\boldsymbol{P})=0$ for all $\boldsymbol{P}\in \mathsf{L}$, and
\item for any $\boldsymbol{P}\in \mathsf{NS}$, $f(\boldsymbol{P})\geq f\left(\mathcal{W}(\boldsymbol{P})\right)$
for all $\mathcal{W}\in\mathsf{LOSR}\ (\mathsf{WPICC})$. 
\end{enumerate}
\end{definition}

Note that, due to the inclusion relation \eqref{eq:LOSR_subset_WPICC}, any $\mathsf{WPICC}$ monotone is automatically also an $\mathsf{LOSR}$ monotone. In the following lemma we show that the converse implication also holds.
\begin{lemma}(Bell nonlocality monotones)
\label{theorem:LOSR_mon_WPICC_mon} 
Let $f:\mathsf{NS}\to\mathbb{R}_{\geq0}$ be an $\mathsf{LOSR}$ monotone, then $f$ is also a $\mathsf{WPICC}$ monotone. 
\end{lemma}

\noindent The \lo{lemma} follows from Eq. \eqref{eq:finalWPICC} \lo{and from the fact that any  convex mixture of a given $\boldsymbol{P}$ with a local behavior (even when the latter is a function of $\boldsymbol{P}$) can always be realized by some specific $\mathsf{WPICC}$ wiring applied on $\boldsymbol{P}$ \cite{deVicente14}. This implies} that for every $\mc{W} \in \mathsf{WPICC}$ and $\boldsymbol{P} \in \mathsf{NS}$, there exists \lo{some} $\mc{W'} \in \mathsf{LOSR}$ such that $\mc{W}(\boldsymbol{P})=\mc{W'}(\boldsymbol{P})$. Hence, \lo{it follows} that
\begin{equation}
f(\mc{W}(\boldsymbol{P})) =f(\mc{W'}(\boldsymbol{P}))\leq f\boldsymbol(P)\lo{,}
\end{equation} 
\lo{where the inequality is due to the $\mathsf{LOSR}$ monotonicity of $f$. The lemma} motivates a unified definition of quantifiers of Bell nonlocality.

\begin{definition}(Bell nonlocality monotones)
\label{def:Bell_NL_monot} 
We call any $\mathsf{LOSR}$ \lo{or $\mathsf{WPICC}$ monotone} a \emph{Bell nonlocality monotone}.
\end{definition}
Next, we construct a Bell nonlocality monotone based on the behavior RE of Def. \ref{dfn:rel_ent_b}.
\subsection{The relative entropy of nonlocality}
One can define the \emph{relative entropy of Bell nonlocality} following the analogous procedure to that used for the REs of entanglement \cite{VedralPlenio} or steering  \cite{Gallego15}, e.g. That is, for any $\boldsymbol{P}\in \mathsf{NS}$, the RE of nonlocality of $\boldsymbol{P}$ is defined as
\begin{equation}
\label{def:rel_ent_nl} 
S_{\rm{nl}}(\boldsymbol{P})\coloneqq\min_{\boldsymbol{P}_\mathsf{L}\in\mathsf{L}}S_{\rm{b}}\left(\boldsymbol{P}\|\boldsymbol{P}_\mathsf{L}\right).
\end{equation}
In other words, $S_{\rm{nl}}(\boldsymbol{P})$ is the minimal behavior RE of $\boldsymbol{P}$ with respect to any local behavior
.  
The RE \eqref{def:rel_ent_nl} was originally introduced in Ref. \cite{vanDam05} as one of the variants of the so-called \emph{statistical strength of nonlocality proofs}, using game-theoretic considerations, as discussed in detail in Sec. \ref{sec:st_strength}. An advantage of defining the RE of nonlocality in terms of the RE between behaviors as in Eq. \eqref{eq:relent_behavior} is that one can prove that $S_{\rm{nl}}$ is a resource-theoretically valid measure of Bell nonlocality with the same recipe as for the REs of entanglement \cite{VedralPlenio} or steering \cite{Gallego15}. This leads to the following result, proven in App. \ref{sec:Proof_theo_mon_RE_BNL}.
\begin{theorem}(Bell monotonicity of the RE of nonlocality)
\label{theo:mon_RE_BNL}
The RE of Bell nonlocality $S_{\rm{nl}}$  is a Bell nonlocality monotone. 
\end{theorem}

 
\section{Connection to the statistical strength of nonlocality proofs}
\label{sec:st_strength}
In Ref. \cite{vanDam05}, van Dam, Grunwald, and Gill (vDGG) introduced an information-theoretic measure of the statistical strength of non-local behaviors as Bell nonlocality proofs, named the \emph{statistical strength of nonlocality proofs}. For any nonlocal behavior $\boldsymbol{P}_\mathsf{NL}\in\mathsf{NS}$, this measure quantifies the minimum statistical confidence that an output sample generated by $\boldsymbol{P}_\mathsf{NL}$ gives in support of the hypothesis HNL that the outputs have indeed been generated from $\boldsymbol{P}_\mathsf{NL}$ and against the hypothesis HL that the data have been produced by any behavior $\boldsymbol{P}_\mathsf{L}\in\mathsf{L}$, when the inputs are chosen according to a probability distribution $\boldsymbol{D}$, and maximising over $\boldsymbol{D}$. vDGG proposed three variants, given by three different constraints on the allowed input distribution $\boldsymbol{D}$:
\begin{subequations}
\label{eqs:stat_strength}
\begin{align}
\label{eq:def_KLDnon1}
S_{\text{u}}(\boldsymbol{P}_\mathsf{NL})
\coloneqq&\min_{\boldsymbol{P}_\mathsf{L} \in \mathsf{L}} S\left(\boldsymbol{P}_\mathsf{NL}\cdot\boldsymbol{D}^{(u)}\|\boldsymbol{P}_\mathsf{L}\cdot\boldsymbol{D}^{(u)}\right),\\
\label{eq:def_KLDnon2}
S_{\text{uc}}(\boldsymbol{P}_\mathsf{NL})
\coloneqq&\max_{\boldsymbol{D}\in \mathsf{UC}}\, \min_{\boldsymbol{P}_\mathsf{L} \in \mathsf{L}}S\left(\boldsymbol{P}_\mathsf{NL}\cdot\boldsymbol{D}\|\boldsymbol{P}_\mathsf{L}\cdot\boldsymbol{D}\right),\\
\label{eq:def_KLDnon3}
S_{\text{c}}(\boldsymbol{P}_\mathsf{NL})
\coloneqq&\max_{\boldsymbol{D}}\, \min_{\boldsymbol{P}_\mathsf{L} \in \mathsf{L}} S\left(\boldsymbol{P}_\mathsf{NL}\cdot\boldsymbol{D}\|\boldsymbol{P}_\mathsf{L}\cdot\boldsymbol{D}\right).
\end{align}
\end{subequations}
The labels u, uc, and c, stand respectively for uniform, uncorrelated, and correlated. $\boldsymbol{D}^{(u)}$ is the uniform probability distribution of settings, with elements $I_\text{u}(x,y)=1/s^2$ for all $x,\,y\in[s]$. The maximisation in Eq. \eqref{eq:def_KLDnon2} is restricted to the subset $\mathsf{UC}\coloneqq\{\boldsymbol{D}: \boldsymbol{D}= \boldsymbol{D}_{X}\cdot\boldsymbol{D}_{Y} \}$ of all probability distributions of uncorrelated settings chosen independently with arbitrary local distributions $\boldsymbol{D}_{X}$ and $\boldsymbol{D}_{Y}$. In contrast, the maximisation in Eq. \eqref{eq:def_KLDnon3} runs over the whole the simplex of all probability distributions $\boldsymbol{D}$, including those for which the inputs are correlated. 

We note, also, that, in Ref. \cite{Horodecki15}, measures analogous to \eqref{eq:def_KLDnon1} and \eqref{eq:def_KLDnon3} have been considered for contextuality.

\subsection{Interpretation of the three variants of the strength of nonlocality proofs}
vDGG interprete the three variants in Eq. \eqref{eqs:stat_strength} in terms of a two-player game. One of the players, QUANTUM, supporter of the hypothesis HNL, wants to convince the other one, CLASSICAL, supporter of  HL, that HNL is true and HL is false. To this end, QUANTUM takes a box with nonlocal behavior $\boldsymbol{P}_\mathsf{NL}\notin\mathsf{L}$, chooses its inputs according to a distribution $\boldsymbol{D}$, and samples outputs from it. The more distinguishable the resulting input-output distribution is from any one generated by a local behavior $\boldsymbol{P}_\mathsf{L}\in\mathsf{L}$, for the same input distribution $\boldsymbol{D}$, the more evident it becomes for CLASSICAL that QUANTUM is right.
In turn, each of the variants in Eq. \eqref{eqs:stat_strength} measures the optimal (over $\boldsymbol{D}$) asymptotic statistical confidence in HNL when QUANTUM is allowed to choose the inputs uniformly,  uncorrelated or  arbitrarily.
Of all three definitions, vDGG favour Eq. \eqref{eq:def_KLDnon2}, corresponding to uncorrected inputs, as the most reasonable one. On the one hand, the authors see no physical reason why QUANTUM should restrict to uniformly chosen inputs to rightfully convince CLASSICAL. On the other one, they argue that using generic (possibly correlated) inputs makes QUANTUM's case weaker, as it could give CLASSICAL the impression that some hidden communication between Alice and Bob might be taking place.

While it is certainly true that choosing inputs uniformly is unnecessarily restrictive, we find the restriction to uncorrelated inputs unnecessary too. As is well known, correlations between Alice and Bob's inputs cannot be used by QUANTUM to fake nonlocality using a local behavior  (see, e.g., Ref. \cite{Tobias2012}. What could, in contrast, be used to cheat CLASSICAL are correlations between the inputs and the hidden variable \cite{Bacon03}. However, such correlations are totally independent of whether the inputs are uniform, uncorrelated, or correlated. Thus, we see no reason why QUANTUM should restrict to independent inputs to rightfully convince CLASSICAL. As for what game-theoretic interpretations concerns, we view Eq. \eqref{eq:def_KLDnon3}, corresponding to generic inputs, as the most reasonable definition of all three. 

On the other hand, from a resource-theoretic perspective, QUANTUM's point should be made based on nonlocality measures in the sense of Defs. \ref{def:monot} and \ref{def:Bell_NL_monot}. That is, the statistical strength of nonlocality proofs should be quantified by a Bell nonlocality monotone. Otherwise, one may run into situations where the statistical strength is ill-defined. To see this, suppose that QUANTUM takes, as nonlocality proof, a given nonlocal behavior $\boldsymbol{P}_\mathsf{NL}\notin\mathsf{L}$. According to Eqs. \eqref{eqs:stat_strength}, its statistical strength of nonlocality should be $S_{\text{i}}(\boldsymbol{P}_\mathsf{NL})$, where the subindex i can represent any of three input-choice strategies. However, before generating the outputs, QUANTUM has the freedom of modifying his nonlocality proof by having Alice and Bob apply a free wiring $\mc{W}$ to it, for instance $\mc{W}\in\mathsf{LOSR}$ or $\mc{W}\in\mathsf{WPICC}$. CLASSICAL cannot complain about this, since, by definition, $\mc{W}$ cannot create nonlocality, but the statistical strength would then  change from $S_{\text{i}}(\boldsymbol{P}_\mathsf{NL})$ to $S_{\text{i}}\left(\mc{W}(\boldsymbol{P}_\mathsf{NL})\right)$, which may be greater than $S_{\text{i}}(\boldsymbol{P}_\mathsf{NL})$. Furthermore, this process could be repeated indefinitely, obtaining every time, according to Eqs. \eqref{eqs:stat_strength}, a different (wiring-dependent) value of the statistical strength. This problem is circumvented if the statistical strength is defined so as to satisfy nonlocality monotonicity. 

\subsection{Monotonicity/non-monotonicity of the strengths of nonlocality}
We next study whether the three variants of the strength of nonlocality proofs given by Eqs. \eqref{eqs:stat_strength} are Bell nonlocality monotones. 
First of all, note that \cite{vanDam05}
\begin{equation}
\label{bnl=c}
S_{\rm{nl}}=S_{\text{c}}
\end{equation}
(see App. \ref{sec:Proof_rel_ent_stat_stre} for an explicit proof). Hence, the monotonicity of $S_{\text{c}}$ follows automatically from Theorem \ref{theo:mon_RE_BNL}.
On the other hand, for $S_{\text{u}}$ and $S_{\text{uc}}$, monotonicity is addressed by the following two lemmas, which we prove in Apps. \ref{sec:Proof_lemma_non_monot} and \ref{proof:lem_monotuc}.
\begin{theorem}(Non-monotonicity of $S_{\text{u}}$)
\label{theor:non_monot} 
$S_{\text{u}}$ \emph{is not} an $\mathsf{LOSR}$ monotone. Hence, it is, in addition, neither a $\mathsf{WPICC}$ monotone nor a Bell nonlocality monotone.
\end{theorem}

\begin{theorem}(Monotonicity of $S_{\text{uc}}$)
\label{theor:monotuc} 
$S_{\text{uc}}$ \emph{is} a Bell nonlocality monotone.
\end{theorem}

The fact that $S_{\text{u}}$ is not monotone under free wirings rules it out as a resource-theoretic consistent candidate for the statistical strength of nonlocality. In contrast, $S_{\text{uc}}$ is consistent with both resource theories of Bell nonlocality, the one based on $\mathsf{LOSR}$ as well as that based on $\mathsf{WPICC}$. Nonetheless, it is worth mentioning that our proof of Thm. \ref{theor:monotuc} directly relies on the fact that the hypothesis HL against one is testing is that $\boldsymbol{P}_\mathsf{NL}\notin\mathsf{L}$ (see App. \ref{proof:lem_monotuc}) and does not work against stronger hypotheses. For instance, imagine the hypothetical situation of a third player, POST-QUANTUM, supporter of the hypothesis HPQ that the data have been generated by a post-quantum no-signaling behavior $\boldsymbol{P}_{\mathsf{PQ}}\in\mathsf{NS}\setminus\mathsf{Q}$ and against the hypothesis HQ that they have been produced by any quantum behavior $\boldsymbol{P}_\mathsf{Q}\in\mathsf{Q}$. POST-QUANTUM wants to convince QUANTUM that quantum theory is violated and offers $\boldsymbol{P}_{\mathsf{PQ}}$ as a post-quantumness proof. The analogous of Eq. \eqref{eq:def_KLDnon2} relevant for such scenario would then be
\begin{equation}
\label{eq:def_KLDnon2Q}
S^{(\mathsf{Q})}_{\text{uc}}(\boldsymbol{P}_{\mathsf{PQ}})
\coloneqq \max_{\boldsymbol{D}\in \mathsf{UC}}\, \min_{\boldsymbol{P}_\mathsf{Q} \in \mathsf{Q}}S\left(\boldsymbol{P}_{\mathsf{PQ}}\cdot\boldsymbol{D}\|\boldsymbol{P}_\mathsf{Q}\cdot\boldsymbol{D}\right).
\end{equation}
Accordingly, monotonicity of should be shown under \emph{quantum wirings}, defined analogously to in Eq. \eqref{eq:LOSR} but with the inputs and outputs wired to boxes in $\mathsf{Q}$ in stead of $\mathsf{L}$.
The proof of monotonicity given in App. \ref{proof:lem_monotuc} (specifically, Lem. \ref{lem:decompositionuclosr} there) does not hold for $S^{(\mathsf{Q})}_{\text{uc}}$ under quantum wirings. We leave as an open question whether $S^{(\mathsf{Q})}_{\text{uc}}$ is a monotone under quantum wirings. 

\section{Discussion}
\label{sec:conclu}
We would like to finish with a few relevant remarks about our results and some open questions. First of all, at first blush, Thm. \ref{theor:notmonotonewpicc} (our central theorem), which proves non-monotonicity of the behavior  relative entropy $S_{\rm{nl}}$ under $\mathsf{WPICC}$ wirings, may give the impression that it could be possible to increase the distinguishability of a given nonlocal behavior from the local ones by a $\mathsf{WPICC}$ wiring. This would be directly relevant, e.g., for nonlocality certification. However, we know from Thm. \ref{theo:mon_RE_BNL}, which proves $\mathsf{WPICC}$ monotonicity of the relative entropy of nonlocality $S_{\rm{nl}}$,
 that such an increase is impossible. Interestingly, $\mathsf{WPICC}$ manages to increase $S_{\rm{b}}$ but within subspaces of bounded $S_{\rm{nl}}$. 
In fact, the two exemplary behaviors $\boldsymbol{P}$ and $\boldsymbol{P}'$ given in the proof of Thm. \ref{theor:notmonotonewpicc}, whose relative entropy increases under a $\mathsf{WPICC}$ wiring, are both local, i.e., both have zero $S_{\rm{nl}}$. While we have not seriously attempted to find  non-local behaviors, or a local and a nonlocal ones, for which $S_{\rm{b}}$ increases under $\mathsf{WPICC}$s, we would not be surprised if such examples were found. A possible search strategy for them could be to consider convex combinations of the exemplary behaviors of the proof of Thm. \ref{theor:notmonotonewpicc}, and similar $\mathsf{WPICC}$ wirings too, and numerically optimise the measurement settings to assess $S_{\rm{b}}$.

Second, as already mentioned, another consequence of Thm. \ref{theor:notmonotonewpicc} is that a \emph{bona fide} measure of the overall distinguishability between behaviors must explicitly take into account that behaviors are multi-partite objects for which some of the inputs can be chosen depending on the outputs of other users. 

Third, concerning the strength of nonlocality proofs, operational consistency demands that its definition incorporates Bell nonlocality monotonicity as a built-in property. We have shown that the variant for which the inputs are sampled from the uniform distribution does not fulfil this. The other two variants introduced in Ref. \cite{vanDam05}, with correlated and uncorrelated inputs, are both Bell nonlocality monotones, although (as an interesting side remark) it is an open question whether the variant with uncorrelated inputs would be consistent with an operational framework with quantum wirings. Either way, from the latter two variants, the fully general one allowing for correlated inputs seems the most appropriate to us. On the one hand, we find the restriction to uncorrelated inputs unnecessary and, on the other one, the variant with correlated inputs coincide with the $S_{\rm{nl}}$ as defined directly in terms of $S_{\rm{b}}$.

Fourth, as we have discussed, since $\mathsf{L}$ is closed under both $\mathsf{LOSR}$ and $\mathsf{WPICC}$, from the mathematical point of view, both classes can be taken as valid free operations for Bell nonlocality.
However, from the physical viewpoint, one class could be more appropriate than the other for some  situation. This depends on the natural physical constraints native of the specific task for which the nonlocal correlations are serving as resource. For instance, since $\mathsf{WPICC}$ allows for communication about inputs and outputs of the initial box, it is legitimate to ask wether $\mathsf{WPICC}$ can indeed be allowed in -- say -- quantum-key distribution. There, depending on the protocol, security constraints might impose $\mathsf{LOSR}$ over $\mathsf{WPICC}$ as the adequate class of harmless operations.

In conclusion, from a fundamental perspective, our work reveals unexpected features of the internal geometry of the set of wirings and characterises the connections between the two resource-theoretic paradigms for Bell nonlocality. From an applied one, in turn, our findings may be relevant to the operational task of distinguishing black-box measurement devices from a restricted set of measurement settings.
This may for instance be the case in nonlocality certification within cryptographic protocols. There, one may be interested in certifying a nonlocal target behavior from the same measurements used in the protocol, so that a potential Eavesdropper does not know which experimental runs are used for the protocol itself and which ones for the certification.

\section{Acknowledgements}

We thank J. I. de Vicente for useful discussions and for bringing \lo{lemma \ref{theorem:LOSR_mon_WPICC_mon}} to our attention. RG acknowledges funding from DFG (GA 2184/2-1) and ERC (TAQ). LA acknowledges financial support from the Brazilian agencies
CNPq, FAPERJ, CAPES, and INCT-IQ.


\appendix

\section{Parametrisation of the class $\mathsf{WPICC}$}
\label{sec:Par_WPICC}
Here we derive the explicit analytic parametrisation of a generic wiring $\mathcal{W}_\mathsf{WPICC}\in\mathsf{WPICC}$, whose simplified form is given in Eq. \eqref{eq:finalWPICC}. To this end, we use a general decision tree. During the preparation phase, the experimenters decide, using shared randomness, if Alice measures first, if Bob measures first, or if no-one measures. For the first two cases, the user that measures first communicates to the other one his/her chosen input ($x$ or $y$) as well his/her obtained output ($a$ or $b$). Then, conditioned on the two communicated dits, the other user decides whether or not to measure. This gives altogether 5 different cases, and a generic $\mathsf{WPICC}$ wiring allows, of course, for probabilistic mixings of all five branches. The branch where no-one measures during the preparation leads, by definition, to an $\mathsf{LOSR}$ overall resulting wiring. For the two branches where both users measure during the preparation (either Alice first and Bob second or vice versa), the users end up the preparation phase with the correlated random dits $a$, $x$, $b$, and $y$. One may be tempted to think that the latter two branches  also result in an overall $\mathsf{LOSR}$ wiring, but this is not the case, as we  show next.

Let us first group the five cases into three main branches and analyse the transformations experienced by $\PPP$, due to both preparation and measurement phases, along each branch: 
\begin{itemize}
\item \emph{Both Alice and Bob measure their initial boxes.}
Suppose Bob measures first. That is, he chooses $y$ according to a single-partite probability distribution $\boldsymbol{D}_Y$. His box thus outputs $b$ with probability $P_0(b|y)$, given by his marginal behavior from $\PPP$. He sends both dits $y$ and $b$ to Alice. She, in turn, chooses $x$ according to a single-partite probability distribution $\boldsymbol{D}_{X|b,y}$, which explicitly depends on $b$ and $y$, and obtains $a$ with conditional probability $P_0(a|x,b,y)$. In general, Alice can also communicate her dits $x$ and $a$ to Bob, so both parties finish the preparation phase with all four dits. Since the initial inputs are already chosen, the resulting behavior is a joint probability distribution $\PPP^{(B\rightarrow A)}$ with outputs only ($a$, $b$, $x$, and $y$), whose elements are:
\begin{eqnarray}
\nonumber 
\PPP^{(B\rightarrow A)}(a,x,b,y)&=&D_{Y}(y)\,P_0(b|y)\, D_{X|b,y}(x)\\
\nonumber 
&\times& P_0(a|x,b,y)\\
&=&D_{Y}(y)\,P_0(a,b|x,y)\, D_{X|b,y}(x).
\end{eqnarray}
Clearly, since it has no inputs, the resulting behavior $\PPP^{(B\rightarrow A)}$ has a local-hidden variable model. 
Finally, in the measurement phase, $\PPP^{(B\rightarrow A)}$ undergoes the most generic $\mathsf{LOSR}$ wiring acting on behaviors without inputs and for which both users know each other`s initial dits. This is explicitly parametrised by $\QQQ^{(B\rightarrow A)}=\mathcal{W}^{(B\rightarrow A)}_\mathsf{LOSR}\left(\PPP^{(B\rightarrow A)}\right)$, with 
\begin{align}
\label{eq:LOSR_A_to_B}
\nonumber
P_\text{f}^{(B\rightarrow A)}(\alpha,\beta\vert  \chi,\psi)\coloneqq&
\sum_{a,b,x,y}O^{(\mathsf{L})}_{a,b,x,y}(\alpha,\beta\vert \chi,\psi)\\
\times&\,\PP^{(B\rightarrow A)}(a,b\vert x,y).
\end{align}
For each $a,b\in[r]$ and $x,y\in[s]$, $\boldsymbol{O}^{(\mathsf{L})}_{a,b,x,y}$ is a local behavior (with respect to $\alpha$ and $\beta$ as outputs and $\chi$ and $\psi$ as inputs) that depends arbitrarily on $a$, $b$, $x$, and $y$, reflecting the fact that both users know each other`s initial dits.
Clearly, since $\PPP^{(B\rightarrow A)}\in\mathsf{L}$, the final behavior $\QQQ^{(B\rightarrow A)}$ is also in $\mathsf{L}$. 
To end up with, if Alice measures first, one obtains the final behavior $\QQQ^{(A\rightarrow B)}\in\mathsf{L}$, defined analogously to $\QQQ^{(B\rightarrow A)}$.

\item \emph{Either Alice or Bob measures her/his initial box.}
Suppose it is Bob who makes the measurement. \lo{Precisely, Fig. \ref{fig:1} b) represents an example of this situation}. In this case, Bob's actions are the same as in the previous branch. Alice, in contrast, does not  measure her device until the measurement phase, but she holds a copy of Bob`s dits $b$ and $y$. The resulting behavior $\PPP^{(B)}$ has a single input ($x$, on Alice`s side) and its elements are:
\begin{align}
P^{(B)}_0(a,b,y\vert x)=D_{Y}(y)\,P_0(a,b|x,y).
\end{align}
Clearly, since only one side has inputs, the resulting behavior $\PPP^{(B)}$ has a local-hidden variable model too. 
Finally, in the measurement phase, $\PPP^{(B)}$ undergoes the most generic $\mathsf{LOSR}$ wiring on behaviors with inputs only on Alice`s side and for which Alice knows the Bob`s initial dits. This is parametrised by $\QQQ^{(B)}=\mathcal{W}^{(B)}_\mathsf{LOSR}\left(\PPP^{(B)}\right)$, with 
\begin{align}
\label{eq:LOSR_B}
\nonumber
P_\text{f}^{(B)}(\alpha,\beta\vert  \chi,\psi)\coloneqq&
\sum_{a,b,x,y}O^{(\mathsf{L})}_{b,y}(\alpha,\beta\vert a,x,\chi,\psi)\\
\times&\,\PP^{(B)}(a,b,y\vert x)\times I^{(A)}_{b,y}(x\vert \chi).
\end{align}
For each $b\in[r]$ and $y\in[s]$, $\boldsymbol{O}^{(\mathsf{L})}_{b,y}$ is a local behavior (with respect to $\alpha$ and $\beta$ as outputs and $a$, $x$, $\chi$, and $\psi$ as inputs) and $\boldsymbol{I}^{(A)}_{b,y}$ is a single-partite behavior. Both depend arbitrarily on $b$ and $y$, reflecting the fact that Alice knows Bob`s initial dits.
Clearly, since $\PPP^{(B)}\in\mathsf{L}$, the final behavior $\QQQ^{(B)}$ is also in $\mathsf{L}$. 
To end up with, if it is instead Alice who measures, one obtains the final behavior $\QQQ^{(A)}\in\mathsf{L}$, defined analogously to $\QQQ^{(B)}$.

\item \emph{None of the parties measures in the preparation phase.}
In this case, both parties apply directly an $\mathsf{LOSR}$ wiring, leading to the final behavior $\QQQ=\mc{W}_{\mathsf{LOSR}}(\PPP)$, with $\mc{W}^{(\rm{None})}_{\mathsf{LOSR}}$ a generic wiring in $\mathsf{LOSR}$. 
\end{itemize}

As mentioned above, probabilistic mixtures of all five cases are admitted in general, leading to the final expression
\begin{widetext}
\begin{eqnarray}
\nonumber 
\mc{W}_{\text{WPICC}}(\PPP)&=&p_{A\rightarrow B}\,\mc{W}^{(A \rightarrow B)}_{\mathsf{LOSR}}(\PPP^{(A\rightarrow B)})+p_{B\rightarrow A}\,\mc{W}^{(B\rightarrow A)}_{\mathsf{LOSR}}(\PPP^{(B\rightarrow A)})\\&+&p_{A}\,
\label{final_final}
\mc{W}^{A }_{\mathsf{LOSR}}(\PPP^{A})+p_{B}\,
\mc{W}^{B}_{\mathsf{LOSR}}(\PPP^{B})+ p_\text{None}\,
\mc{W}^{(\rm{None})}(\PPP),
\end{eqnarray}
with $p_{A\rightarrow B}, p_{B\rightarrow A}, p_{A}, p_{B}, p_\text{None}\geq0$ and $p_{A\rightarrow B}+p_{B\rightarrow A}+p_{A}+p_{B}+ p_\text{None}=1$. Making, next, in Eq. \eqref{final_final}, the identifications $p\coloneqq 1-p^{(\text{None})}$ and 
\begin{eqnarray}
\label{eq:formulaL}
L(\PPP)&\coloneqq&\frac{1}{p}\left[p_{A\rightarrow B}\,\mc{W}^{(A \rightarrow B)}_{\mathsf{LOSR}}(\PPP^{(A\rightarrow B)})+p_{B\rightarrow A}\,\mc{W}^{(B\rightarrow A)}_{\mathsf{LOSR}}(\PPP^{(B\rightarrow A)})+p_{A}\,
\mc{W}^{A }_{\mathsf{LOSR}}(\PPP^{A})+p_{B}\,
\mc{W}^{B}_{\mathsf{LOSR}}(\PPP^{B})\right],
\end{eqnarray}
\end{widetext}
one arrives at Eq. \eqref{eq:finalWPICC}. Note that $L(\PPP)\in\mathsf{L}$ for all $\PPP\in\mathsf{NS}$, because each term in Eq. \eqref{eq:formulaL} is a behavior in $\mathsf{L}$.
\section{Proof of Thm. \ref{theor:relentmonotone}}
\label{sec:Proof_theor:relentmonotone}
Note first that all wirings $\mc{W}_{\mathsf{GW}}\in\mathsf{GW}$ treat the bipartite boxes $\boldsymbol{P}$ and $\boldsymbol{P}'$ as if they were monopartite boxes with single-partite inputs $i=(x,y)$ and output $o=(a,b)$. Similarly, we also make the identifications $\phi=(\chi,\psi)$ and $\gamma=(\alpha,\beta)$. With this notation, we must show that Eq. \eqref{eq:relentmonotone} holds using that $\QQQ=\mathcal{W}_\mathsf{GW}\left(\PPP\right)$, with 
\begin{align}
\label{eq:GO_new_not}
P_\text{f}(\gamma\vert\phi):=
\sum_{o,i}O(\gamma\vert o,i,\phi)\,\PP(o\vert i)\, I(i\vert \phi).
\end{align}

\begin{widetext}
Using Eqs. \eqref{eq:def_KLD} and \eqref{eq:relent_behavior}, we write
\begin{eqnarray}
\nonumber S_{\rm{b}}(\mathcal{W}_{\mathsf{GW}}(\boldsymbol{P})|\mathcal{W}_{\mathsf{GW}}(\boldsymbol{P}'))&=&\max_{\phi}\sum_{\gamma}P_f(\gamma|\phi)\log \left(\frac{P_f(\gamma|\phi)}{P'_f(\gamma|\phi)}\right)\\
\label{eq:proofmonotone1}&=&\max_{\phi}\sum_{\gamma}
\sum_{o,i}O(\gamma\vert o,i,\phi)\,\PP(o\vert i)\, I(i\vert \phi)
\log \left(\frac{\sum_{o,i}O(\gamma\vert o,i,\phi)\,\PP(o\vert i)\, I(i\vert \phi)}{\sum_{o,i}O(\gamma\vert o,i,\phi)\,\PP'(o\vert i)\, I(i\vert \phi)}\right)\\
\label{eq:proofmonotone2}&\leq &\max_{\phi}\sum_{\gamma}
\sum_{o,i}O(\gamma\vert o,i,\phi)\,\PP(o\vert i)\, I(i\vert \phi)
\log \left(\frac{O(\gamma\vert o,i,\phi)\,\PP(o\vert i)\, I(i\vert \phi)}{O(\gamma\vert o,i,\phi)\,\PP'(o\vert i)\, I(i\vert \phi)}\right)\\
\label{eq:proofmonotone3}&=&\max_{\phi}\sum_{\gamma}
\sum_{o,i}O(\gamma\vert o,i,\phi)\,\PP(o\vert i)\, I(i\vert \phi)
\log \left(\frac{\PP(o\vert i)}{\PP'(o\vert i)}\right)\\
\label{eq:proofmonotone4}&=&\max_{\phi}
\sum_{o,i}\PP(o\vert i)\, I(i\vert \phi)
\log \left(\frac{\PP(o\vert i)}{\PP'(o\vert i)}\right)\\
\label{eq:proofmonotone5}&=&\max_{\phi}\sum_{i}\mc{W}_{\text{i}}(i|\phi)\,S(\boldsymbol{P}(\cdot,i)\|\boldsymbol{P}'(\cdot,i))\\
\label{eq:proofmonotone6}&\leq &\max_{i}S(\boldsymbol{P}(\cdot,i)\|\boldsymbol{P}'(\cdot,i))\\
\label{eq:proofmonotone7}&=&S_{\rm{b}}(\boldsymbol{P}\|\boldsymbol{P}'),
\end{eqnarray}
where \eqref{eq:proofmonotone1} follows from \eqref{eq:GO_new_not}; \eqref{eq:proofmonotone2}  from the well-known property that $\sum_{i} x_i \log (\sum_i x_i /\sum_i y_i) \leq \sum_{i} x_i \log (x_i / y_i)$ if $x_i\geq0$ and $y_i \geq 0$ $\forall i$; \eqref{eq:proofmonotone3} from basic algebra; \eqref{eq:proofmonotone4} from summing over $\gamma$ and using that the probability distributions are normalized; \eqref{eq:proofmonotone5} from Eq. \eqref{eq:def_KLD}; \eqref{eq:proofmonotone6}  from the fact that the average is smaller than the largest value; and \eqref{eq:proofmonotone7} from Eq. \eqref{eq:relent_behavior}.
\end{widetext}


\section{Proof of Thm. \ref{theor:notmonotonewpicc}}
\label{sec:Proof_theor:notmonotonewpicc}
Our proof strategy consists of constructing two concrete behaviors $\PPP,\,\PPP'\in\mathsf{NS}$ and a concrete wiring $\mathcal{W}_\mathsf{WPICC} \in \text{WPICC}$ such that Eq. \eqref{eq:violatemonotonicity} is fulfilled. In fact, our construction takes place in the simple scenario of $r_A=r_B=s_A=2$ and $s_B=1$. That is, we consider a single input for Bob. Explicitly, $\PPP=\{P_0(a,b|x)\}_{a,b\in[2],\, x\in[2]}$ and $\PPP'=\{P'_0(a,b|x)\}_{a,b\in[2],\, x\in[2]}$. For $0< \epsilon < \frac{1}{2}$, we choose the components of $\PPP$ and $\PPP'$ as:
\begin{center}
\begin{tabular}{ |c || c | c |c  }
   \hline
   $P_0(a,b|x)$ & $x=0$ & $x=1$ \\ \hline                    
  $a=0,b=0$ & $\frac{1}{2} - \epsilon$ & $\frac{1}{2} - \epsilon$ \\ \hline
  $a=1,b=0$ & $\epsilon$ & $\epsilon$ \\ \hline
  $a=0,b=1$ & $\epsilon$ & $\epsilon$ \\ \hline
  $a=1,b=1$ & $\frac{1}{2} - \epsilon$ & $\frac{1}{2} - \epsilon$ \\  \hline
  \end{tabular}
  \end{center}
  
  \begin{center}
\begin{tabular}{ |c || c | c |c  }
   \hline
   $P'_0(a,b|x)$ & $x=0$ & $x=1$ \\ \hline                      
  $a=0,b=0$ & $\epsilon$ & $\frac{1}{2} - \epsilon$\\ \hline
  $a=1,b=0$ & $\frac{1}{2} - \epsilon$ & $\epsilon$ \\ \hline
  $a=0,b=1$ & $\epsilon$ & $\frac{1}{2} - \epsilon$ \\ \hline
  $a=1,b=1$ & $\frac{1}{2} - \epsilon$ & $\epsilon$ \\  \hline
  \end{tabular}
  \end{center}
  
As the reader can immediately verify, both distributions are well-normalized. In addition, $P_0(b|x=0)=P_0(b|x=1)$ and $P'_0(b|x=0)=P'_0(b|x=1)$, which implies that the behaviors are no-signaling. This, together with the fact that Bob has a single input, implies that , so that $\PPP$ and $\PPP$ are actually in $\mathsf{L}$.

Let us first give some intuition of why one expect that \eqref{eq:violatemonotonicity} can hold for these distributions. $S_{\rm{b}} (\PPP\|\PPP')=\max_{x} S\left(\PPP(\cdot,\cdot|x)\|\PPP'(\cdot,\cdot|x)\right)$ measures the distinguishability between the output distributions resulting from $\PPP$ and $\PPP'$ when the input is fixed at $x$, maximised over $x$. Note that, due to the fact that $\PPP$ is independent of $x$ and the symmetries of $\PPP'$, they are equally indistinguishable for $x=0$ and $x=1$. Let us analyse how distinguishable they are for each input value: For $x=0$, when $b=1$ the resulting distributions over $a$ are the same, whereas they are different for $b=0$ (see the first column of the tables). In turn, for $x=1$, when $b=0$ the resulting distributions over $a$ are the same, whereas they are different for $b=1$ (see the second column of the tables). 
Now, consider the following wiring $\mathcal{W}_\mathsf{WPICC} \in \text{WPICC}$: In the preparation phase, Bob presses the only button of his initial box (he has no choice of settings for this, since $s_b=1$), obtains the bit $b$ as output, and sends it to Alice. In the measurement phase, in turn, Alice receives the input $\chi$ of her final box, ignores it, and chooses the input to her initial box as $x=b$. There are then no further wirings and the final outputs are simply $\alpha=a$ and $\beta=b$. By doing this, Alice and Bob avoid the cases $(x=0,b=1)$ and $(x=1,b=0)$ for which $\PPP$ and $\PPP'$ behave the same. Heuristically, we expect that this transformation should increase the distinguishability. Let us next prove it rigorously. 

First, we re-write the right-hand side of \eqref{eq:violatemonotonicity} as
\begin{eqnarray}
\nonumber 
S_{\rm{b}}\left(\PPP\|\PPP'\right) &=& \max_{x}S\left(\PPP(\cdot,\cdot|x)\|\PPP'(\cdot,\cdot|x)\right)\\
\nonumber
 &=&S\left(\PPP(\cdot,\cdot|0)\|\PPP'(\cdot,\cdot|0)\right)\\
\nonumber
 &=&\sum_{a,\,b}P_0(a,b\vert 0)\log \left( \frac{P_0(a,b\vert 0)}{P'_0(a,b\vert 0)} \right)\\
\label{eq:S_b_initial}
 &=&\left( \frac{1}{2} -2 \epsilon \right) \log \left( \frac{\frac{1}{2}-\epsilon}{\epsilon} \right).
\end{eqnarray} 
The above-mentioned wiring is such that the number of inputs and outputs is preserved, i.e. ${r_A}_f={r_B}_f={s_A}_f=2$ and ${s_B}_f=1$. The final behaviors are in turn given by $\QQQ=\mathcal{W}_\mathsf{GW}\left(\PPP\right)$ and $\QQQ'=\mathcal{W}_\mathsf{GW}\left(\PPP'\right)$, with
\begin{subequations}
\label{eq:final_wiring}
\begin{align}
P_f(a,b\vert\chi) =&P_0(a, b\vert b)\ \forall\ a,b,\chi\in[2],\\
P'_f(a,b\vert\chi) =&P'_0(a, b\vert b)\ \forall\ a,b,\chi\in[2].
\end{align}
\end{subequations}
We emphasise that both final behaviors are independent of $\chi$, since Alice chooses the input to her initial box as $x=b$ ignoring the value of the final input $\chi$. Explicitly, the resulting components of the final behaviors are given by the following tables.
\begin{center}
\begin{tabular}{ |c || c | c |c  }
   \hline
   $\QQ(a,b|\chi)$ & $\chi=0$ & $\chi=1$ \\ \hline                    
  $a=0,b=0$ & $\frac{1}{2} - \epsilon$ & $\frac{1}{2} - \epsilon$ \\ \hline
  $a=1,b=0$ & $\epsilon$ & $\epsilon$ \\ \hline
  $a=0,b=1$ & $\epsilon$ & $\epsilon$ \\ \hline
  $a=1,b=1$ & $\frac{1}{2} - \epsilon$ & $\frac{1}{2} - \epsilon$ \\  \hline
  \end{tabular}
  \end{center}
  
  \begin{center}
\begin{tabular}{ |c || c | c |c  }
   \hline
   $\QQ'(a,b|\chi)$ & $\chi=0$ & $\chi=1$ \\ \hline                      
  $a=0,b=0$ & $\epsilon$ &$\epsilon$\\ \hline
  $a=1,b=0$ & $\frac{1}{2} - \epsilon$ & $\frac{1}{2} - \epsilon$  \\ \hline
  $a=0,b=1$ & $\frac{1}{2} - \epsilon$ &   $\frac{1}{2} - \epsilon$\\ \hline
  $a=1,b=1$ & $\epsilon$ & $\epsilon$  \\  \hline
  \end{tabular}
  \end{center}

Comparing either of the two columns of the tables (both columns are equal, since they are independent of $\chi$), it is clear that they are more distinguishable that either of the two the columns of the initial tables. Indeed, using Eqs. \eqref{eq:final_wiring}, we compute the left-hand side of \eqref{eq:violatemonotonicity}:
\begin{eqnarray}
\nonumber 
S_{\rm{b}}\left(\QQQ\|\QQQ'\right) &=& \max_{\chi}S\left(\QQQ(\cdot,\cdot|\chi)\|\QQQ'(\cdot,\cdot|\chi)\right)\\
\nonumber
 &=&S\left(\QQQ(\cdot,\cdot|0)\|\QQQ'(\cdot,\cdot|0)\right)\\
\nonumber
 &=&\sum_{a,\,b}\QQ(a,b\vert 0)\log \left( \frac{\QQ(a,b\vert 0)}{\QQ'(a,b\vert 0)} \right)\\
 \nonumber
 &=&\sum_{a,\,b}\PP(a, b\vert b)\log \left( \frac{\PP(a, b\vert b)}{\PP'(a,b\vert b)} \right)\\
 \nonumber
 &=&2\,\left( \frac{1}{2} -2 \epsilon \right) \log \left( \frac{\frac{1}{2}-\epsilon}{\epsilon} \right)\\
 &=&2\, S_{\rm{b}}\left(\PPP\|\PPP'\right),
\end{eqnarray}
with $S_{\rm{b}}\left(\PPP\|\PPP'\right)$ given by Eq. \eqref{eq:S_b_initial}. As we see, the RE between the final behaviors doubles the RE between the initial ones.
%
%

As a final remark, we note that similar examples can be found for behaviors with larger alphabets. Interestingly, there, the RE can increase unboundedly with the size of the alphabets.

\section{Proof of Thm. \ref{theo:mon_RE_BNL}}
\label{sec:Proof_theo_mon_RE_BNL}

Let us first see that $S_{\rm{nl}}$ is a $\mathsf{LOSR}$ monotone. This follows from Thm. \ref{theor:relentmonotone}, which implies that $S_{\rm{b}}$ is contractive under $\mathsf{LOSR}$ wirings, together with Def. \ref{def:monot} and property \eqref{eq:free_ops}.
Then, we show that  $S_{\rm{nl}}$ is convex, which, by virtue of Lemma \ref{theorem:LOSR_mon_WPICC_mon}, implies that  $S_{\rm{nl}}$ is also a $\mathsf{WPICC}$ monotone. Convexity of  $S_{\rm{nl}}$ follows straightforwardly from the convexity of $S_{\rm{b}}$. We show this explicitly in the following lemma.

\begin{lemma}(Convexity of $S_{\rm{nl}}$)
\label{lem:convexity_Snl} 
The relative entropy of nonlocality $S_{\rm{nl}}$ is convex. That is,
\begin{align}
\label{eq:convexitySnl}
S_{\rm{nl}}\left(\mu\, \boldsymbol{P}+(1-\mu)\boldsymbol{P}'\right)&\leq \mu\, S_{\rm{nl}}\left(\boldsymbol{P}\right)+(1-\mu)\,S_{\rm{nl}}\left(\boldsymbol{P}'\right),
\end{align}
for all $0\leq\mu\leq 1$ and all $\boldsymbol{P},\,\boldsymbol{P}'\in\mathsf{NS}$.
\end{lemma}
\begin{proof}
Let $\boldsymbol{\bar{P}}^*=\boldsymbol{\bar{P}}^*(\mu)$ be such that $S_{\rm{nl}}\left(\mu\, \boldsymbol{P}+(1-\mu)\,\boldsymbol{P}'\right)=S_{\rm{b}}\left(\mu\, \boldsymbol{P}+(1-\mu)\boldsymbol{P}'\|\boldsymbol{\bar{P}}^*\right)$, i.e., the optimal local behavior minimising Eq. \eqref{def:rel_ent_nl}. Equivalently, let $\boldsymbol{P^*}$ and $\boldsymbol{P'}^*$ be the optimal local behavior minimizing $S_{\rm{nl}}\left(\boldsymbol{P}\right)$ and $S_{\rm{nl}}\left(\boldsymbol{P}'\right)$ respectively.  Then
\begin{eqnarray}
\nonumber
&&S_{\rm{nl}}\left(\mu\, \boldsymbol{P}+(1-\mu)\boldsymbol{P}'\right)\\
\nonumber 
&\coloneqq &S_{\rm{b}}\left(\mu\, \boldsymbol{P}+(1-\mu)\,\boldsymbol{P}' \|\boldsymbol{\bar{P}}^*\right)\\
\nonumber 
&\leq &S_{\rm{b}}\left(\mu\, \boldsymbol{P}+(1-\mu)\,\boldsymbol{P}' \|\mu \boldsymbol{P}^*+ (1-\mu) \boldsymbol{P'}^*\right)\\
\nonumber 
&\leq&\max_{x,y\in[s]}\Big[\mu\, S\left(\boldsymbol{P}(\cdot,\cdot\vert x,y)\|\boldsymbol{P}^*(\cdot,\cdot\vert x,y)\right)\\
\label{eq:convexity1}
&+&(1-\mu)\,S\left(\boldsymbol{P}'(\cdot,\cdot\vert x,y)\|\boldsymbol{P'}^*(\cdot,\cdot\vert x,y)\right)\Big]\\
\label{eq:convexity2}
&\leq&\mu\, S_{\rm{b}}(\boldsymbol{P}\|\boldsymbol{P}^*)+(1-\mu)\,S_{\rm{b}}(\boldsymbol{P}'  \|\boldsymbol{P'}^*)\\
\nonumber &\eqqcolon &  \mu\, S_{\rm{nl}}\left(\boldsymbol{P}\right)+(1-\mu)\,S_{\rm{nl}}\left(\boldsymbol{P}'\right)
\end{eqnarray}
where \eqref{eq:convexity1} follows from Eq. \eqref{eq:relent_behavior} and joint-convexity of $S$ and \eqref{eq:convexity2} from  Eq. \eqref{eq:relent_behavior} and from maximising the values of $x$ and $y$ for each term independently.
\end{proof}

\section{Proof that $S_{\rm{nl}}=S_{\text{c}}$}
\label{sec:Proof_rel_ent_stat_stre}
Note that, since the simplex of generic bipartite probability distributions $\boldsymbol{I}$ is a convex set, von-Neumman minimax theorem implies that the order of the maximisation and the minimisation in Eq. \eqref{eq:def_KLDnon3} is irrelevant \cite{vanDam05}. That is, for any $\boldsymbol{P}_\mathsf{NL}\in\mathsf{NL}$,
\begin{align}
\label{eq:invert_order}
\nonumber
S_{\text{c}}(\boldsymbol{P}_\mathsf{NL})\coloneqq&\max_{\boldsymbol{I}}\, \min_{\boldsymbol{P}_\mathsf{L} \in \mathsf{L}} S\left(\boldsymbol{P}_\mathsf{NL}\cdot\boldsymbol{D}\|\boldsymbol{P}_\mathsf{L}\cdot\boldsymbol{D}\right)\\
\nonumber
=&\min_{\boldsymbol{P}_\mathsf{L} \in \mathsf{L}}\, \max_{\boldsymbol{I}} S\left(\boldsymbol{P}_\mathsf{NL}\cdot\boldsymbol{D}\|\boldsymbol{P}_\mathsf{L}\cdot\boldsymbol{D}\right)\\
=&\min_{\boldsymbol{P}_\mathsf{L} \in \mathsf{L}}S_{\rm{b}}\left(\boldsymbol{P}_\mathsf{NL}\|\boldsymbol{P}_\mathsf{L}\right),
\end{align}
where Eq. \eqref{eq:relent_behavior} has been used in the last equality. As evident from Eq. \eqref{def:rel_ent_nl}, the last term is precisely $S_{\rm{nl}}(\boldsymbol{P}_\mathsf{NL})$.
\section{Proof of Thm. \ref{theor:non_monot}}
\label{sec:Proof_lemma_non_monot}
The proof is by construction. That is, we find a concrete behavior $\PPP\in\mathsf{Q}$ and a concrete wiring $\mathcal{W}_\mathsf{LOSR}\in\mathsf{LOSR}$, such that $S_{\rm{u}}\left(\QQQ\right)>S_{\rm{u}}(\PPP)$, where $\QQQ=\mathcal{W}_\mathsf{LOSR}(\PPP)$. 

To this end, let us first consider the Bell scenario where $s=s_f=4$ and $r=r_f=2$, i.e., four inputs and two outputs for both initial and final behaviors. There, we take $\PPP$ as equal to the so-called Tsirelson box, for input values $x, y \in \{0,1\}$, and to the white-noise uniform distribution, for input values $x, y \in \{2,3\}$. That is, the components of  $\PPP$  are
\begin{center}
\begin{tabular}{ |c || c | c |c | c }
   \hline
   $P_0(a,b\vert x,y)$ & $x\times y=0$ & $x\times y=1$ & $x\times y>1$ \\ \hline                    
  $a=0,b=0$ & $\frac{p}{2}$  & $\frac{1-p}{2}$ & $\frac{1}{4}$\\ \hline
  $a=1,b=0$ & $\frac{1-p}{2}$ & $\frac{p}{2}$ & $\frac{1}{4}$\\ \hline
  $a=0,b=1$ & $\frac{1-p}{2}$ & $\frac{p}{2}$ &$\frac{1}{4}$\\ \hline
  $a=1,b=1$ & $\frac{p}{2}$ & $\frac{1-p}{2}$ &$\frac{1}{4}$\\  \hline
  \end{tabular}\, ,
  \end{center}  
where $p=\frac{1}{2}+\frac{1}{2\sqrt{2}}$. A possible physical realisation of this behavior is to have Alice and Bob perform adequate quantum measurements on a maximally entangled 2-qubit state for two inputs and simply output a random bit for the other two inputs.
%

Then, as our exemplary wiring, we consider $\mathcal{W}_\mathsf{LOSR}$ given by Eq. \eqref{eq:LOSR} with the input box chosen as
\begin{equation}
I^{(\mathsf{L})}(x,y\vert \chi,\psi) = \left\{
     \begin{array}{lr}
       \delta(x,\chi)\,\delta(y,\psi) &:  \text{if }\chi\times\psi\leq 1,\\
       \delta\left(x,|\chi|_2\right)\,\delta\left(y,|\psi|_2\right) &:  \text{if }\chi\times\psi>1,
     \end{array}
   \right.
\end{equation}
where $\delta$ stands for the Kronecker delta and $|\ |_2$ for modulo 2, and the output box chosen as $O^{(\mathsf{L})}(\alpha,\beta\vert a,b,x,y,\chi,\psi)=\delta(\alpha,a)\,\delta(\beta,b)$ for all $x$, $y$, $\chi$, and $\psi$. Applying this map to $\PPP$ gives the final behavior $\QQQ$, of elements:
\begin{center}
\begin{tabular}{ |c || c | c |}
   \hline
   $P_f(\alpha,\beta|\chi,\psi)$ & $\chi,\psi\in \mathcal{S}_1$ & $\chi,\psi \in \mathcal{S}_2$  \\ \hline                    
  $\alpha=0,\beta=0$ & $\frac{p}{2}$  & $\frac{1-p}{2}$ \\ \hline
  $\alpha=1,\beta=0$ & $\frac{1-p}{2}$ & $\frac{p}{2}$  \\ \hline
  $\alpha=0,\beta=1$ & $\frac{1-p}{2}$ & $\frac{p}{2}$ \\ \hline
  $\alpha=1,\beta=1$ & $\frac{p}{2}$ & $\frac{1-p}{2}$ \\  \hline
  \end{tabular}\, ,
  \end{center}
where the setting sets $\mathcal{S}_1\coloneqq\{(0,0), (0,1),(0,2),(0,3),$ $(1,0),(1,2),(2,0),(2,1),(2,2),(2,3),(3,0),(3,2)\}$ and $\mathcal{S}_2\coloneqq\{(1,1),(1,3),(3,1),(3,3)\}$ have been introduced. Note that 
\begin{equation}
\label{eq:final_equal_initial}
\QQQ(\cdot,\cdot|\chi,\psi)=\PPP(\cdot,\cdot|\chi,\psi)\ \forall\ \chi,\psi \text{ s. t. } \chi\times\psi\leq 1.
\end{equation}

Next, we find a lower bound to $S_{\rm{u}}(\QQQ)$ that will be seen to be greater than $S_{\rm{u}}(\PPP)$ below. To this end, we note, using Eq. \eqref{eq:def_KLDnon1}, that
\begin{eqnarray}
\nonumber
S_{\rm{u}}(\QQQ)&=&\\
\nonumber
\min_{\boldsymbol{P}_\mathsf{L}\in \mathsf{L}} \frac{1}{16} \sum_{\chi,\psi} S\left(\QQQ(\cdot,\cdot|\chi,\psi)\|\boldsymbol{P}_\mathsf{L}(\cdot,\cdot|\chi,\psi)\right)&=&\\
\nonumber
\frac{1}{16}\min_{\boldsymbol{P}_\mathsf{L}\in \mathsf{L}}\sum_{i=1}^4\sum_{(\chi,\psi) \in \mathcal{K}_i}  S\left(\QQQ(\cdot,\cdot|\chi,\psi)\|\boldsymbol{P}_\mathsf{L}(\cdot,\cdot|\chi,\psi)\right)&\geq&\\
\label{eq:example1}
 \frac{1}{16}\sum_{i=1}^4\min_{\boldsymbol{P}^{(i)}_\mathsf{L}\in \mathsf{L}}\sum_{(\chi,\psi) \in \mathcal{K}_i}  S\left(\QQQ(\cdot,\cdot|\chi,\psi)\|\boldsymbol{P}^{(i)}_\mathsf{L}(\cdot,\cdot|\chi,\psi)\right),
\end{eqnarray}
where $\mathcal{K}_1\coloneqq\{ (0,0), (0,1), (1,0),\textbf{(1,1)}\}$, $\mathcal{K}_2\coloneqq\{(0,2),(0,3),(1,2),\textbf{(1,3)}\}$, $\mathcal{K}_3\coloneqq\{(2,0),(2,1),(3,0),\textbf{(3,1)}\}$ and $\mathcal{K}_4\coloneqq\{(2,2),(2,3),(3,2),\textbf{(3,3)}\}$. The four measurement settings in bold are the ones belonging to the set $\mathcal{S}_2$ above. All other settings belong to $\mathcal{S}_1$. This, by inspection of the table for $P_f(\alpha,\beta|\chi,\psi)$ above, implies that the four terms in \eqref{eq:example1} are actually equal. Hence, 
\begin{align}
\nonumber
S_{\rm{u}}(\QQQ)&\geq& \frac{1}{4} \min_{\boldsymbol{P}_\mathsf{L}\in \mathsf{L}} \sum_{(\chi,\psi) \in \mathcal{K}_1} S\left(\QQQ(\cdot,\cdot|\chi,\psi)\|\boldsymbol{P}_\mathsf{L}(\cdot,\cdot|\chi,\psi)\right)\\
\label{eq:mintsi}
&=& \frac{1}{4} \min_{\boldsymbol{P}_\mathsf{L}\in \mathsf{L}} \sum_{\chi\times\psi\leq 1} S\left(\QQQ(\cdot,\cdot|\chi,\psi)\|\boldsymbol{P}_\mathsf{L}(\cdot,\cdot|\chi,\psi)\right).
\end{align}

Finally, we find an upper bound to $S_{\rm{u}}(\PPP)$ that is smaller than the right-hand-side of  Eq. \eqref{eq:mintsi}.  For this, it is convenient to introduce a behavior $\tilde{\boldsymbol{P}}\in \mathsf{L}$, in the restricted $r=2=s$ scenario, such that 
\begin{align}
\nonumber
\label{eq:S_U_opt_restrict}
\sum_{x\times y\leq1} S\left(\PPP(\cdot,\cdot\vert x,y)\|\tilde{\boldsymbol{P}}(\cdot,\cdot\vert x,y)\right)=&\\
\min_{\boldsymbol{P}_\mathsf{L}\in \mathsf{L}} \sum_{x\times y\leq1} S\left(\PPP(\cdot,\cdot\vert x,y)\|\boldsymbol{P}_\mathsf{L}(\cdot,\cdot\vert x,y)\right)&.
\end{align}
That is, $\tilde{\boldsymbol{P}}\in \mathsf{L}$ minimises the (uniform input) RE with respect to $\PPP$ for the restricted subset of inputs $(x,y)\in \mathcal{K}_1$. With this, we construct an educated guess $\boldsymbol{P}^*\in \mathsf{L}$ of the actual optimal local behavior 
for $\PPP$ for generic inputs, i.e., the one attaining the minimisation that defines $S_{\rm{u}}(\PPP)$. 
We define $\boldsymbol{P}^*$ so as to coincide with $\tilde{\boldsymbol{P}}$ over the restricted input set $\mathcal{K}_1$ and with the white-noise uniform distribution over the other inputs. That is,
\begin{equation}
P^*(a,b\vert x,y)\coloneqq \left\{
     \begin{array}{lr}
      \tilde{P}(a,b\vert x,y) &:  \text{if }x\times y\leq 1,\\
       1/4 &:  \text{if }x\times y>1,
     \end{array}
   \right.
\end{equation}
for all $a,b\in[2]$. Then, it holds that
\begin{eqnarray}
\label{eq:eqxe0}
\nonumber
&&S_{\rm{u}}(\PPP) \coloneqq\frac{1}{16}  \min_{\boldsymbol{P}_\mathsf{L}\in \mathsf{L}} \sum_{x,y} S\left(\PPP(\cdot,\cdot\vert x,y)\|\boldsymbol{P}_\mathsf{L}(\cdot,\cdot\vert x,y)\right)\\
\nonumber 
&\leq& \frac{1}{16} \sum_{x,y} S\left(\PPP(\cdot,\cdot\vert x,y)\|\boldsymbol{P}^*(\cdot,\cdot\vert x,y)\right)\\
\nonumber 
&=& \frac{1}{16} \Bigg[ \sum_{x\times y\leq 1} S\left(\PPP(\cdot,\cdot\vert x,y)\|\boldsymbol{P}^*(\cdot,\cdot\vert x,y)\right)\\
\nonumber 
&+&\sum_{x\times y> 1}S\left(\PPP(\cdot,\cdot\vert x,y)\|\boldsymbol{P}^*(\cdot,\cdot\vert x,y)\right)\Bigg]\\
\label{eq:eqxe1}
&=&\frac{1}{16}\sum_{x\times y\leq 1} S\left(\PPP(\cdot,\cdot\vert x,y)\|\boldsymbol{P}^*(\cdot,\cdot\vert x,y)\right)\\
\label{eq:eqxe2}
&=&\frac{1}{16}\sum_{x\times y\leq 1} S\left(\PPP(\cdot,\cdot\vert x,y)\|\tilde{\boldsymbol{P}}(\cdot,\cdot\vert x,y)\right)\\
\label{eq:eqxe3}
&=&\frac{1}{16}\min_{\boldsymbol{P}_\mathsf{L}\in \mathsf{L}} \sum_{x\times y\leq1} S\left(\PPP(\cdot,\cdot\vert x,y)\|\boldsymbol{P}_\mathsf{L}(\cdot,\cdot\vert x,y)\right),
\end{eqnarray}
where 
equality \eqref{eq:eqxe1} follows from the fact that $\PPP$ and $\boldsymbol{P}^*$ coincide when $x\times y>1$ (they are both equal to the flat distribution), \eqref{eq:eqxe2} from the fact that, by construction of $\boldsymbol{P}^*$, $\boldsymbol{P}^*$ and $\tilde{\boldsymbol{P}}$ coincide over the restricted input set $\mathcal{K}_1$, and equality \eqref{eq:eqxe3} from the definition of $\tilde{\boldsymbol{P}}$ in Eq. \eqref{eq:S_U_opt_restrict}.

The last step is simply to note, using Eqs. \eqref{eq:final_equal_initial}, \eqref{eq:eqxe3}, and \eqref{eq:mintsi}, that 
\begin{equation}
S_{\rm{u}}(\QQQ)\geq 4\,S_{\rm{u}}(\PPP)>S_{\rm{u}}(\PPP),
\end{equation}
which shows the theorem's claim.

\section{Proof of Thm. \ref{theor:monotuc}}
\label{proof:lem_monotuc}
Here, we show that $S_{\text{uc}}$, as defined in Eq. \eqref{eq:def_KLDnon3}, is a Bell nonlocality monotone. To this end, we first show that it is a $\mathsf{LOSR}$ monotone and then that it is convex. By virtue of Thm. \ref{theorem:LOSR_mon_WPICC_mon} and Def. \ref{def:Bell_NL_monot}, these two facts prove the theorem.

To show monotonicity, it is convenient to introduce a subclass of $\mathsf{LOSR}$ wirings called \emph{uncorrelated local operations assisted by shared randomness}, denoted as $\mathsf{UCLOSR}\subset\mathsf{LOSR}$. It is composed of all wirings $\mathcal{W}_\mathsf{UCLOSR}$ explicitly parametrised by  $\QQQ=\mathcal{W}_\mathsf{UCLOSR}\left(\PPP\right)$, with 
\begin{align}
\label{eq:UC-LOSR}
\nonumber
P_\text{f}(\alpha,\beta\vert  \chi,\psi):=
\sum_{a,b,x,y}&O^{(\mathsf{UC})}(\alpha,\beta\vert a,b,x,y,\chi,\psi)\times\\
&\PP(a,b\vert x,y)\times I^{(\mathsf{UC})}(x,y\vert \chi,\psi),
\end{align}
where $\boldsymbol{I}^{(\mathsf{UC})}\coloneqq\boldsymbol{I}^{(A)}\cdot\boldsymbol{I}^{(B)}_\text{i}$ and $\boldsymbol{O}^{(\mathsf{UC})}\coloneqq\boldsymbol{O}^{(A)}\cdot\boldsymbol{O}^{(B)}$ are uncorrelated local boxes composed of independent local behaviors for Alice and Bob:
\begin{eqnarray}
\nonumber
\boldsymbol{I}^{(A)}&\coloneqq& \{I^{(A)}(x\vert\chi)\}_{x\in[s],\, \chi\in[s_f]},\\
\nonumber
\boldsymbol{O}^{(A)}&\coloneqq& \{O^{(A)}(\alpha\vert a,x,\chi)\}_{\alpha\in[r_f],\,a\in[r],\,x\in[s],\, \chi\in[s_f]},\\
\nonumber
\boldsymbol{I}^{(B)}&\coloneqq& \{I^{(B)}(y\vert\psi)\}_{y\in[s],\, \psi\in[s_f]},\\
\nonumber
\boldsymbol{O}^{(B)}&\coloneqq& \{O^{(B)}(\beta\vert b,\psi, y)\}_{\beta\in[r_f],\,b\in[r],\,y\in[s],\, \psi\in[S_f]}.
\end{eqnarray}
Then, the following fact is true.
\begin{lemma} ($\mathsf{UCLOSR}$ monotonicity of $S_{\rm{uc}}$)
\label{lem:uclosrmon}
Let $\boldsymbol{P}\in\mathsf{NS}$ be any  no-signaling behavior. Then
\begin{equation}
\label{eq:uclosrmon}
S_{\rm{uc}}\left(\mathcal{W}_\mathsf{UCLOSR}(\boldsymbol{P})\right) \leq S_{\rm{uc}}\left(\boldsymbol{P}\right)
\end{equation}
for all $\mc{W}_\mathsf{UCLOSR}\in\mathsf{UCLOSR}$.
\end{lemma}
The proof of this lemma is completely analogous to the one of Thm. \ref{theor:relentmonotone}. 

In addition, using Eqs. \eqref{eq:LHV} and \eqref{eq:LOSR}, one immediately proves the following fact.
\begin{lemma}
\label{lem:decompositionuclosr} 
Every $\mathcal{W}_\mathsf{LOSR}\in \mathsf{LOSR}$ can be decomposed as
\begin{equation}
\mathcal{W}_\mathsf{LOSR}=\sum_{\lambda} p(\lambda) \mc{W}^{(\lambda)}_{\mathsf{UCLOSR}},
\end{equation}
where $\mc{W}^{(\lambda)}_{\mathsf{UCLOSR}}\in\mathsf{UCLOSR}$ for all $\lambda$.
\end{lemma}
Note that, if $S_{\rm{uc}}$, lemmas \ref{lem:uclosrmon} and \ref{lem:decompositionuclosr} together imply $\mathsf{LOSR}$ monotonicity of $S_{\rm{uc}}$. So, the only missing ingredient is to show convexity of $S_{\rm{uc}}$, which we do next.
\begin{lemma}(Convexity of $S_{\rm{uc}}$)
\label{lem:convexity} 
The statistical strength $S_{\rm{uc}}$ is convex. That is,
\begin{align}
\label{eq:convexitySuc}
S_{\rm{uc}}\left(\mu\, \boldsymbol{P}+(1-\mu)\boldsymbol{P}'\right)&\leq \mu\, S_{\rm{uc}}\left(\boldsymbol{P}\right)+(1-\mu)\,S_{\rm{uc}}\left(\boldsymbol{P}'\right),
\end{align}
for all $0\leq\mu\leq 1$ and all $\boldsymbol{P},\,\boldsymbol{P}'\in\mathsf{NS}$.
\end{lemma}
The proof of this lemma is totally analogous to that of Lem. \ref{lem:convexity_Snl}.

\begin{thebibliography}{99}
\bibitem{comment1}
We note that,  in Ref. \cite{deVicente14} , the terminology wirings and prior-to-input classical communication is used for the class that we call $\mathsf{LOSR}$. The ambiguity originates at the fact that, until now, it was unclear if there could be any difference between both classes of wirings. 
\end{thebibliography}
\end{document}